\documentclass[10pt,twocolumn,SingleSpace]{IEEEtran}
\usepackage{setspace}

\usepackage{amsmath}
\usepackage{amsthm}
\usepackage{pxfonts}
\usepackage{cite}
\usepackage{algorithmic}
\usepackage{algorithm}
\usepackage{epsfig}
\usepackage{color}
\usepackage{enumerate}
\usepackage{subfigure}
\usepackage{graphicx}
\graphicspath{{figures/}}

\usepackage{cite, citesort}
\usepackage{times}
\usepackage{multirow}
\renewcommand{\algorithmicrequire}{\textbf{Input:}} 
\renewcommand{\algorithmicensure}{\textbf{Output:}} 

\theoremstyle{plain}

\newtheorem{theorem}{Theorem}[section]
\newtheorem{problem}{Problem}

\newtheorem{lemma}[theorem]{Lemma}
\newtheorem{fact}{Fact}

\begin{document}

\title{Merge Frame Design for Video Stream Switching \\
using Piecewise Constant Functions}

\author{Wei Dai~\IEEEmembership{Student Member,~IEEE},
Gene Cheung~\IEEEmembership{Senior Member,~IEEE},
Ngai-Man Cheung~\IEEEmembership{Senior Member,~IEEE},\\
Antonio Ortega~\IEEEmembership{Fellow,~IEEE},
Oscar C. Au~\IEEEmembership{Fellow,~IEEE}
\begin{small}
\thanks{W. Dai is with Hong Kong University of Science and Technology, Clear Water Bay, Kowloon, Hong Kong. Email: weidai@connect.ust.hk}
\thanks{G. Cheung is with National Institute of Informatics, 2-1-2, Hitotsubashi, Chiyoda-ku, Tokyo, 101-8430, Japan. Email: cheung@nii.ac.jp}
\thanks{N.-M. Cheung is with Singapore University of Technology and Design, 8 Somapah Road, Singapore 487372. Email: ngaiman\_cheung@sutd.edu.sg}
\thanks{A. Ortega is with University of Southern California, 3740 McClintock Ave., Los Angeles, CA 90089-2564. Email: ortega@sipi.usc.edu}
\end{small}
}
\maketitle

\markboth{IEEE Transactions on Image Processing, September~2015}%
{Dai \MakeLowercase{\textit{et al.}}: Merge Frame Design for Video Stream Switching using Piecewise Constant Functions}



\begin{abstract}
  The ability to efficiently switch from one pre-encoded video stream
  to another (e.g., for bitrate adaptation or view switching) is
  important for many interactive streaming applications. Recently,
  stream-switching mechanisms based on distributed source coding (DSC)
  have been proposed. In order to reduce the overall transmission
  rate, these approaches provide a ``merge'' mechanism, where
  information is sent to the decoder such that the exact same frame can 
  be reconstructed given that any one of a known set of
  side information (SI) frames is available at the decoder (e.g., each
  SI frame may correspond to a different stream from which we are
  switching). However, the use of bit-plane coding and channel coding
  in many DSC approaches leads to complex coding and decoding.  In
  this paper, we propose an alternative approach for merging multiple
  SI frames, using a piecewise constant (PWC) function as the merge
  operator. In our approach, for each block to be reconstructed, a
  series of parameters of these PWC merge functions are transmitted in
  order to guarantee identical reconstruction given the known side
  information blocks. We consider two different scenarios. In the
  first case, a target frame is first given, and then
  merge parameters are chosen so that this frame can be reconstructed
  exactly at the decoder. In contrast, in the second scenario, the   
  reconstructed frame and merge parameters are jointly optimized to meet a
  rate-distortion criteria. Experiments show that for both scenarios,
  our proposed merge techniques can outperform both a recent approach
  based on DSC and the SP-frame approach in H.264, in terms of
  compression efficiency and decoder complexity.
\end{abstract}


\IEEEpeerreviewmaketitle

\section{Introduction}
\label{sec:intro}

In conventional \textit{non-interactive} video streaming, a client plays back 
successive frames in a pre-encoded stream in a fixed order. In contrast, in \textit{interactive} video streaming~\cite{cheung10vcip}, a client can 
switch freely in real-time among a number of pre-encoded streams. 
Examples include switching among multiple streams representing the same video encoded at different bit-rates for real-time bandwidth adaptation~\cite{wu08}, or switching among views in a multi-view video~\cite{cheung11tip}. 
See \cite{cheung10vcip} for more examples of interactive streaming.
A major challenge in interactive video streaming is to achieve efficient real-time switching among pre-encoded video streams. A simple approach would be to insert an intra-coded I-frame at each potential switching point~\cite{cheung09pv}. But the relatively high rate required for I-frames often makes it impractical to insert them frequently in the streams, thus reducing the interactivity of playback.   

\begin{figure}[ht]
\centering
\includegraphics[width = 0.4\textwidth]{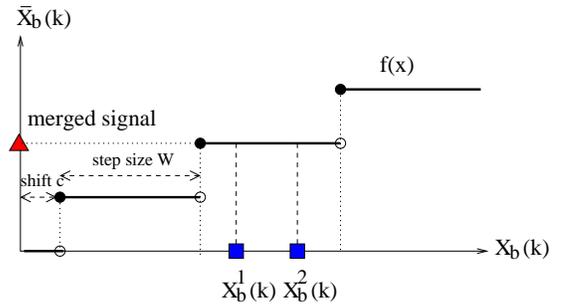}
\vspace{-0.1in}
\caption{Given the $k$-th coefficient $X_b(k)$ in block $b$ from either SI frame 1 or 2, a piecewise constant function $f(x)$ maps either one ($X^1_b(k)$ or $X^2_b(k)$) to the same $\bar{X}_b(k)$ if they fall on the same constant interval.}
\label{fig:pwc}
\end{figure}

Towards a more efficient stream-switching mechanism, \textit{distributed source coding} (DSC) has been proposed. DSC can in principle achieve compression efficiency that is a function of the worst-case correlation between the target frame and the \textit{side information} (SI) frames (from which the client may be switching) \cite{nmcheung06,mcheung08vcip,mcheung09pcs}. As an example, illustrated by Fig.\;\ref{fig:pwc}, in the block-based DCT approach of \cite{mcheung09pcs}, a desired $k$-th 
quantized frequency coefficient value $\bar{X}_b(k)$ in block $b$ of the target frame is reconstructed using either $X^1_b(k)$ or $X^2_b(k)$, the corresponding coefficients in SI frames $1$ and $2$, respectively. A \textit{D-frame} is transmitted so that it is possible to reconstruct the exact same target frame given any one of the two SI frames \cite{mcheung09pcs}. Thus we say that the D-frame supports a \textit{merge} operation. 
In particular, the least significant bits (LSBs) of $X^1_b(k)$ and $X^2_b(k)$ are treated as ``noisy'' versions of the LSBs of $\bar{X}_b(k)$. The most significant bits (MSBs) of $\bar{X}_b(k)$ are obtained from the MSBs of $X^1_b(k)$ or  $X^2_b(k)$, which are identical, while the D-frame contains channel codes that can produce the actual LSBs of $\bar{X}_b(k)$ taking $X^1_b(k)$ or  $X^2_b(k)$ as inputs. 
The channel codes associated to these target frame coefficients compose the D-frames, which potentially require significantly fewer bits than an I-frame representation of the target frame~\cite{mcheung09pcs}.
%
%

There remain significant hurdles towards practical implementation of D-frames, however. First, the use of bit-plane encoding and channel codes in proposed techniques~\cite{mcheung09pcs}  means that the computation complexity at the decoder is high. Second, because the average statistics of a transform coefficient bit-plane for the entire image are used, non-stationary noise statistics can lead to high rate channel codes, resulting in coding inefficiency.

In this paper, we propose to use a \textit{piecewise constant} (PWC) function\footnote{An earlier version of this paper was presented at ICIP 2013~\cite{dai13}.} as the signal merging operator. This approach operates directly on quantized frequency coefficients (instead of using a bit-plane representation) and does not require channel codes. As will be discussed in more detail in Section~\ref{subsec:coset}, our signal merging approach can be interpreted as a generalization of \textit{coset coding}~\cite{pradhan:03}, where we explicitly optimize the merged target values for improved rate-distortion (RD) performance. The basic idea of our approach is  
summarized in Fig.\;\ref{fig:pwc}, which depicts a \texttt{floor} function characterized by two parameters: a step size $W$ and a shift $c$. 
In our approach, the encoder selects $W$ and $c$ to guarantee that $X^1_b(k)$ and $X^2_b(k)$ are in the same interval and thus map to the same reconstruction value. A $W$ will be chosen for each frequency $k$, based on the statistics of the various $X_b(k)$ across all blocks $b$. Then, given $W$ it will be possible to adjust $c$ so that the reconstructed value matches a desired target, $\bar{X}_b(k)$. A value of $c$ will be chosen for each $k$ and $b$, so that the bitrate required by our proposed technique is dominated by the cost of transmitting $c$. In this paper, we will formulate the problem of selecting $c$ and $W$, and develop techniques for RD optimization of this selection.  

We consider two scenarios. In the first one, \textit{fixed target merging}, we will assume that $\bar{X}_b(k)$ has been given, \textit{e.g.}, by first generating an intra-coded version of the target frame, and using the corresponding quantized coefficient values as
targets. We will show how to choose $W$ to guarantee that $\bar{X}_b(k)$ can be reconstructed. We will also show that given $W$, $c$ is fixed. 
This type of merging is useful when there are cycles in the interactive playback, \textit{i.e.}, frame $A$ is an SI frame for frame $B$ {\em and} $B$ is an SI frame for $A$. 
This will be the case in \textit{static view switching} for multiview video streaming, to be discussed in Section~\ref{sec:system}. 

In the second scenario, \textit{optimized target merging}, we select $W$, $c$ {\em and} $\bar{X}_b(k)$ based on an RD criteria, where distortion is computed with respect to a desired target $X^0_b(k)$. In this scenario, we can use smaller values for $W$, and no longer need to select a fixed $c$ for a given $W$ and $\bar{X}_b(k)$. This allows us to optimize $c$ so as to significantly reduce the rate needed to encode the merging information. This approach can be used when there are no cycles in the interactive playback, \textit{e.g.}, in \textit{dynamic view switching} scenarios (also discussed in Section \ref{sec:system}). Experimental results show significant compression gains over D-frames \cite{mcheung09pcs} and SP-frames in H.264 \cite{karczewicz03} at reduced decoder computation complexity.

The paper is organized as follows. We first summarize related work in Section~\ref{sec:related}. We then provide an overview of our coding system in Section~\ref{sec:system}. We discuss the use of PWC functions for signal merging in Section~\ref{sec:formulation}. We present our PWC function parameter selection methods for fixed target merging and optimized target merging in Section~\ref{sec:target} and \ref{sec:Solving}, respectively. Finally, we present experimental results and conclusions in Section~\ref{sec:results} and \ref{sec:conclude}, respectively.

\section{Related Work}
\label{sec:related}
The H.264 video coding standard \cite{wiegand03} introduced the concept of \textit{SP-frames}~\cite{karczewicz03} for stream-switching. In a nutshell, first the difference between one SI frame and the target picture is \textit{lossily} coded as the primary SP-frame. Then, the difference between each additional SI frame and the reconstructed primary SP-frame is \textit{losslessly} coded as a secondary SP-frame; lossless coding ensures identical reconstruction between primary and each of the secondary SP-frames. One drawback of SP-frames is coding inefficiency.  Due to lossless coding in secondary SP-frames, their sizes can be significantly larger than conventional P-frames.  Furthermore, the number of secondary SP-frames required is equal to the number of SI frames, thus resulting in significant storage costs. As we will discuss, our proposed scheme encodes only one merge frame for all SI frames, and hence the storage requirement is lower than for SP-frames. 

While DSC has been proposed for designing interactive and stream-switching mechanisms in the past decade~\cite{wu08,aaron04,nmcheung06,mcheung08vcip,mcheung09pcs}, partly due to the computation complexity required for bit-plane and channel coding in common DSC implementations, DSC is not widely used nor adopted into any video coding standards. In contrast, in this work, our proposed coding tool involves only quantization (PWC function) and entropy coding of function parameters, both of which are computationally simple. Further, we demonstrate coding gain over a previously proposed DSC-based approach \cite{mcheung09pcs} in Section~\ref{sec:results}.


One of the primary applications of our proposed merge frame is interactive media systems, which have attracted considerable interest~\cite{cheungCheung:13}. In particular, a range of media data types have been considered for interactive applications in the past: images~\cite{taubman:03}, light-fields~\cite{steinbach:08a,steinbach:08b}, volumetric images~\cite{ortega:09}, videos~\cite{wee:99,lin:01,fu:06,nmcheung06,mcheung08vcip,devaux:07,taubman:11} and high-resolution videos~\cite{girod:11,Mavlankar:10,Halawa:11,Pang:11}. 
While it is conceivable that our proposed merge frame can be applicable in some of these use scenarios for which DSC techniques have been proposed, here we focus on real-time switching among multiple pre-encoded video streams, as discussed in Section \ref{sec:system}.


This paper extends our earlier work~\cite{dai13}, by providing a more detailed presentation and evaluation of the system, as well as introducing two new concepts. First, we study the fixed target merging case (Section \ref{sec:target}). Second, for the optimized target merging case, we develop a new algorithm to compute a locally optimal probability function $P(c)$ for shift $c$---one that leads to more efficient entropy coding of $c$, \textit{and} small signal reconstruction distortion after merging (Section \ref{sec:Solving}). We will show in our experiments, described in Section \ref{sec:results}, that our new algorithm leads to significantly better RD performance than our previously published work~\cite{dai13}.

\section{System Overview}
\label{sec:system}
\subsection{IVSS System Overview}

\begin{figure}[ht]
\centering
\includegraphics[width = 0.45\textwidth]{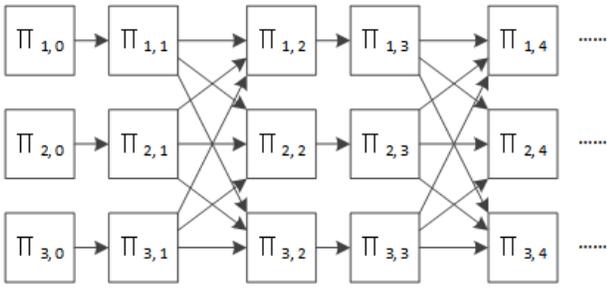}
\caption{Example of an acyclic picture interactivity graph for {\em dynamic view switching}. Each picture $\Pi_{v,t}$ has subscript indicating its view index $v$ and time instant $t$. After viewing picture $\Pi_{2,1}$ of stream 2, the client can choose to keep watching the same stream and jump to $\Pi_{2,2}$, or switch to $\Pi_{1,2}$ or $\Pi_{3,2}$ of stream 1 and 3, respectively.}
\label{fig:DVS}
\end{figure}

We provide an overview of our proposed coding system for \textit{interactive video stream switching} (IVSS), in which our proposed \textit{merge frame} is a key enabling component. In the sequel, a ``picture'' is a raw captured image in a video sequence, while a ``frame'' is a particular coded version of the picture (\textit{e.g.}, I-frame, P-frame). In this terminology, a ``picture'' can have multiple coded versions or ``frames''. 

In an IVSS system, there are multiple pre-encoded video streams that are related (\textit{e.g.}, videos capturing the same 3D scene from different viewpoints \cite{cheung11tip}). During video playback of a single stream, at a \textit{switch instant}, the client can switch from a picture of the original stream to a picture of a different destination stream. Fig.\;\ref{fig:DVS} illustrates an example \textit{picture interactivity graph} for three streams, where there is a switch instant every two pictures in time. An arrow $\Pi_p \rightarrow \Pi_q$ indicates that a switch is possible from picture $\Pi_p$ to picture $\Pi_q$. This particular graph is \textit{acyclic}, {\em i.e.}, it has no loops and we cannot have both $\Pi_p \rightarrow \Pi_q$ and 
$\Pi_q \rightarrow \Pi_p$. 



\begin{figure}[ht]
\centering
\includegraphics[width = 0.4\textwidth]{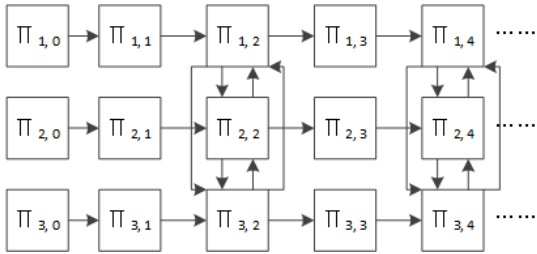}
\caption{Example of a cyclic picture interactivity graph for {\em static view switching}. Each picture $\Pi_{v,t}$ has subscript indicating its view index $v$ and time instant $t$. After viewing $\Pi_{2,2}$ of stream 2, the client can choose to keep watching stream 2 in time and jump to $\Pi_{2,3}$, or change to $\Pi_{1,2}$ or $\Pi_{3,2}$ of stream 1 and 3, respectively, corresponding to the same time instant as $\Pi_{2,2}$.}
\label{fig:SVS}
\end{figure}

The scenario in Fig.\;\ref{fig:DVS} is an example of \textit{dynamic view switching}~\cite{huang12}, where a frame at time $t$ is always followed by a frame at time $t+1$. 
In contrast, in \textit{static view switching} a user can stop temporal playback and interactively select the angle from which to observe a 3D scene frozen in time \cite{lou05}. Fig.\;\ref{fig:SVS} shows an example of static view switching, where the corresponding graph is \textit{cyclic}, \textit{i.e.}, it contains loops so that we can have both $\Pi_p \rightarrow \Pi_q$ and $\Pi_q \rightarrow \Pi_p$. We will discuss the merge frame design for the cyclic case in Section~\ref{sec:target}.


\subsection{Stream-Switch Mechanism in IVSS}

At a given switch instant, stream switching works as follows. First, for each possible switch $\Pi_p \rightarrow \Pi_q$, we encode a P-frame $P_{q\,|\,p}$ for $\Pi_q$, where a decoded version of $\Pi_p$ is used as a predictor. Reconstructed $P_{q\,|\,p}$ is called a \textit{side information} (SI) frame, which constitutes a particular reconstruction of destination $\Pi_q$. Because there are in general multiple origins for a given destination (the \textit{in-degree} for destination picture in the picture interactivity graph), there are multiple corresponding SI frames. Having multiple reconstructions of the same picture $\Pi_q$ creates a problem for the following frame(s) that use $\Pi_q$ as a predictor for predictive coding, because one does not know \textit{a priori} which reconstructed SI frame $P_{q\,|\,p}$ will be available at the decoder buffer for prediction. This illustrates the need for our proposed merge frame (called \textit{M-frame} in the sequel) $M_q$, which is an \textit{extra} frame corresponding to destination $\Pi_q$. Correct decoding of $M_q$ means a unique reconstruction of $\Pi_q$, no matter which SI frame $P_{q\,|\,p}$ is actually available at the decoder.

\begin{figure}[ht]
\centering
\includegraphics[width = 0.4\textwidth]{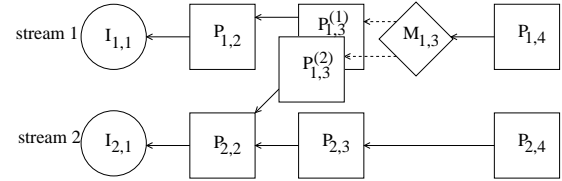}
\caption{Example of stream-switching from one pre-encoded stream to another using merge frame. SI frames $P_{1,3}^{(1)}$ and $P_{1,3}^{(2)}$ are first constructed using predictors $P_{1,2}$ and $P_{2,2}$, respectively. M-frame $M_{1,3}$ is encoded using the two SI frames. I-, P- and M-frames are represented as circles, squares and diamonds, respectively.}
\label{fig:usecase}
\end{figure}

As an illustration, in Fig.\;\ref{fig:usecase} two P-frames, $P_{1,3}^{(1)}$ and $P_{1,3}^{(2)}$,
 generated from predictors $P_{1,2}$ and $P_{2,2}$ respectively, are the SI frames. An M-frame $M_{1,3}$ is added to merge the SI frames to produce an identical reconstruction for $\Pi_{1,3}$. During a stream-switch, the server can transmit any one of the two SI frames \textit{and} $M_{1,3}$ leading to the same reconstructed frame for $\Pi_{1,3}$, thus avoiding coding drift in the following frame $P_{1,4}$.
Note that one P-frame and one M-frame are sent. An alternative approach based on SP frames would require sending a primary SP-frame $S^1_{1,3}$ (using $P_{1,2}$ as the predictor) for the switch $\Pi_{1,2} \rightarrow \Pi_{1,3}$, or a losslessly coded secondary SP-frame $S^2_{1,3}$ (using $P_{2,2}$ as the predictor) for the switch $\Pi_{2,2} \rightarrow \Pi_{1,3}$. SP-frame approaches are asymmetric; rate is much lower when only a primary SP-frame is needed. In contrast, the switching cost using M-frame is always the same (P- and M-frames are transmitted). As will be shown, a combination of a P-frame and an M-frame requires lower rate than a secondary SP-frame. 

\subsection{Merge Frame Overview}

In our proposed M-frame, each fixed-size code block in an SI frame is first transformed to the DCT domain. DCT coefficients are then quantized. The quantized coefficients across SI frames (called \textit{q-coeffs} for short in the sequel) are then examined. If the q-coeffs of a given block are very different across SI frames, then the overhead to merge their differences to targeted q-coeffs would be large. Thus, we will encode the block as a conventional intra block. On the other hand, if the q-coeffs of a given block are already identical across all SI frames, then we can simply inform the decoder that the q-coeffs can be used without further processing. Finally, if the q-coeffs across SI frames are not identical but are similar, then each q-coeff is then merged identically to a target value via our proposed merge operator. Hence, together there are three coding modes for each code block: \textit{intra}, \textit{skip} and \textit{merge}. In this paper, we focus our attention on optimizing the parameters in \textit{merge} mode as the \textit{intra} and \textit{skip} modes are straightforward.




\section{Problem Formulation}
\label{sec:formulation}
\begin{table}
\caption{Table of Notations}
\label{table:notation}
\vspace{-0.15in}
\begin{center}
\begin{footnotesize}
\begin{tabular}{||c|c||} \hline
$N$ & number of SI frames \\ \hline 
$\mathbf{S}^n$ & SI frame $n$ \\ \hline
$\mathbf{T}$ & desired target frame \\ \hline
$\mathbf{M}$ & M-frame \\ \hline
$R(\mathbf{M})$ & rate of M-frame $\mathbf{M}$ \\ \hline
$D(\mathbf{T}, \bar{\mathbf{T}}(\mathbf{M}))$ & distortion of reconstructed $\mathbf{M}$ wrt $\mathbf{T}$ \\ \hline
$\lambda$ & weight parameter to trade off distortion with rate \\ \hline
$\mathcal{B}_M$ & block group encoded in merge mode \\ \hline
$K$ & number of pixels in a code block \\ \hline
$\mathbf{x}^n_b$ & block $b$ of SI frame $\mathbf{S}^n$ \\ \hline
$Y^n_b(k)$ & $k$-th DCT coefficient of block $b$ of SI frame $\mathbf{S}^n$ \\ \hline
$X^n_b(k)$ & $k$-th q-coeff of block $b$ of SI frame $\mathbf{S}^n$ \\ \hline
$Q$ & quantization step size \\ \hline
$\bar{X}_b(k)$ & $k$-th reconstructed q-coeff of block $b$ \\ \hline
$Z^*_b(k)$ & max. pair difference between any pair of $X^n_b(k)$ \\ \hline
$Z^*_{\mathcal{B}_M}(k)$ & group-wise max. pair difference, \textit{i.e.} $\max_{b \in \mathcal{B}_M} Z^*_b(k)$ \\ \hline
$W_{\mathcal{B}_M}(k)$ & step size for $k$-th q-coeff of block group $\mathcal{B}_M$ \\ \hline
$c_b(k)$ & shift parameter for $k$-th q-coeff of block $b$ \\ \hline
$\mathcal{F}_b(k)$ & feasible range of shift $c_b$ for identical merging \\ \hline
$Z_b(k)$ & max. target diff. between target $X^0_b(k)$ and any $X^n_b(k)$ \\ \hline
$Z_{\mathcal{B}_M}(k)$ & group-wise max. target difference, \textit{i.e.} $\max_{b \in \mathcal{B}_M} Z_b(k)$ \\ \hline 
$W^{\#}_{\mathcal{B}_M}(k)$ & step size for $k$-th q-coeff for fixed target merging \\ \hline
\end{tabular}
\end{footnotesize}
\end{center}
\end{table}

\subsection{Notation}
\label{subsec:notation}

We first define the notation that will be used in the sequel; see Table~\ref{table:notation} for quick reference. We denote the $N$ SI frames by $\mathbf{S}^{1}, \ldots, \mathbf{S}^{N}$, one of which is guaranteed to be available at the decoder buffer when M-frame $\mathbf{M}$ is decoded. 
We denote a desired target picture by $\mathbf{T}$ and for notational convenience we will include it in the set of SI frames as  $\mathbf{S}^{0} = \mathbf{T}$.

We denote the group of fixed-size code blocks in $\mathbf{M}$ that are encoded in merge mode by $\mathcal{B}_M$. Each block has $K$ pixels. We denote by $\mathbf{x}^{n}_b$ the $b$-th block in SI frame $\mathbf{S}^n$ coded in merge mode. Each block $\mathbf{x}^{n}_b$ is transformed into the DCT domain as $\mathbf{Y}^{n}_b = [Y^{n}_b(0), \ldots, Y^{n}_b(K-1)]$, where $Y^{n}_b(k)$ is the $k$-th DCT coefficient of $\mathbf{x}^{n}_b$. We denote by $X^{n}_b(k)$ the $k$-th quantized coefficient (\textit{q-coeff}) given uniform quantization step size $Q$: 
\begin{equation}\label{eq:XYQ}
    X^n_b(k) = \mathrm{round} \left( \frac{Y^n_b(k)}{Q} \right) ,
\end{equation}
where $\mathrm{round}(x)$ is the standard rounding operation to the nearest integer.

\subsection{Formulation}
\label{subsec:formulation}

We consider two different problems based on the reconstruction requirement with respect to the desired target $\mathbf{T}$. One typically chooses $\mathbf{T}$ {\em a priori}, \textit{e.g.}, by encoding the target picture independently (intra only) and using the decoded version as $\mathbf{T}$. The first problem requires the M-frame to reconstruct \textit{identically} to desired target $\mathbf{T}$:
\begin{problem}{{\bf Fixed Target Merging}} \label{prob:fixed} (Section~\ref{sec:target}). Find M-frame $\mathbf{M}$ such that the decoder, taking as input \textit{any} one of the SI frames  $\mathbf{S}^{n}$ and $\mathbf{M}$, can reconstruct $\mathbf{T}$ identically as output.  
\end{problem}

Because of the differences between SI frames $\mathbf{S}^n$ and desired target $\mathbf{T}$, there may be situations where a high rate is required for $\mathbf{M}$ (\textit{e.g.}, due to motion in the video sequence, the target frame is very different from previously transmitted frames). In this case, we allow the reconstruction to deviate from desired target $\mathbf{T}$ in order to reduce the rate required for $\mathbf{M}$ by optimizing a rate-distortion criterion: 
\begin{problem}{{\bf Optimized Target Merging}}\label{prob:opt} (Section~\ref{sec:Solving}). Find $\mathbf{M}^*$ and $\bar{\mathbf{T}}(\mathbf{M}^*)$ so that the decoder, taking as input \textit{any} one of SI frames $\mathbf{S}^{n}$ and $\mathbf{M}^*$, can always reconstruct $\bar{\mathbf{T}}(\mathbf{M}^*)$ as output, and where $\mathbf{M}^*$ is an RD-optimal solution for a given weight parameter $\lambda$, \textit{i.e.},
\begin{equation}\label{eq:formular}
\mathbf{M}^* = 
\arg \min_{\mathbf{M}} D(\mathbf{T}, \bar{\mathbf{T}}(\mathbf{M})) + \lambda R(\mathbf{M}) ,
\end{equation}
where $D(\mathbf{T}, \bar{\mathbf{T}}(\mathbf{M}))$ is the distortion incurred (with respect to $\mathbf{T}$) when choosing $\bar{\mathbf{T}}(\mathbf{M})$ as the common reconstructed frame, and $R(\mathbf{M})$ is the rate needed to transmit $\mathbf{M}$. 
\end{problem}

The second problem essentially states that the \textit{reconstruction target} $\bar{\mathbf{T}}(\mathbf{M})$ is RD-optimized with respect to desired target $\mathbf{T}$, while the first problem requires identical reconstruction to desired target $\mathbf{T}$. Note that in both problem formulations we avoid coding drift since they guarantee identical reconstruction for any SI frame, but a solution to Problem~\ref{prob:opt} will be shown to lead to significantly lower coding rates.

\subsection{Piecewise Constant Function for Single Merging}
\label{subsec:PWC}

A merge operation must, given q-coeff $X^n_b(k)$ of any SI frames $\mathbf{S}^n$, $n \in \{1, \ldots, N\}$, reconstruct an identical value $\bar{X}_b(k)$, for all frequencies $k$. 
We use a PWC function $f(x)$ as the chosen merging operator, with {\em shift} $c$ and {\em step size} $W$ parameters selected for each frequency $k$ of each block $b$ encoded in merge mode (see Fig.~\ref{fig:pwc}). 
The selection of these parameters influences the RD performance of this merging operation for the optimized target merging case.  We now focus our discussion on how $c$ and $W$ are selected for each coefficient. Because the optimization is the same for each frequency $k$, we will drop the frequency index $k$ for simplicity of presentation.


Examples of PWC functions are \texttt{ceiling}, \texttt{round}, \texttt{floor}, etc. In this paper, we employ the \texttt{floor} function\footnote{We define \texttt{floor} function to minimize the maximum difference between original $x$ and reconstructed $f(x)$, given shift $c$ and step size $W$.}:
\begin{equation}
f(x) = \left\lfloor \frac{x+c}{W} \right\rfloor W + \frac{W}{2} - c.
\label{eq:floor}
\end{equation}
From Fig.\;\ref{fig:pwc}, it is clear that there are numerous combinations of parameters $W$ and $c$ such that identical merging is ensured---\textit{i.e.}, all $X^n_b$ map to the same constant interval. 
Note also that the choice of $W$ depends on how spread out the various
$X^0_b, \ldots, X^N_b$ are, that is, how correlated the SI blocks are
to each other. In contrast, $c$ is used to select a desired
reconstruction value $X^0_b$.  Thus, because the level of correlation
can be assumed to be relatively consistent across blocks, a
\textit{step size $W_{\mathcal{B}_M}$ is selected once for all blocks
  $b \in \mathcal{B}_M$ for a given frequency}. On the other hand, since the actual reconstruction value will be different from block to block, the {\em
shift $c_b$ will be selected on a per block basis for a given frequency}.
%

Before formulating the problem of optimizing the choice of $c$ and $W$, we derive constraints under which this selection is made by determining: 
\begin{itemize}
\item The minimum value of $W$ that guarantees identical merging, 
\item The choice of $c$ that guarantees correct reconstruction, 
\item Effective range of $c$. 
\end{itemize}


We first compute a minimum step size $W$ to enable identical merging for blocks $b$ in $\mathcal{B}_M$. Let $Z_b^*$ be the \textit{maximum pair difference} between any pair of q-coeffs of a given frequency in block $b$, \textit{i.e.},
\begin{equation}\label{eq:WX}
Z_b^* = \max_{i,j \in \{0, \ldots, N\}} X^i_b - X^j_b 
= X^{\max}_b - X^{\min}_b ,
\end{equation}
where $X^{\max}_b$ and $X^{\min}_b$ are respectively the maximum and minimum q-coeffs among the SI frames, \textit{i.e.},
\begin{equation}
    X^{\max}_b = \max_{n = 0, \ldots, N} X^n_b, ~~~~~
    X^{\min}_b = \min_{n = 0, \ldots, N} X^n_b.
\end{equation}
Given $Z_b^*$, we next define the \textit{group-wise maximum pair difference} $Z_{\mathcal{B}_M}^*$ for the blocks in group $\mathcal{B}_M$:
\begin{equation}
Z_{\mathcal{B}_M}^* = \max_{b \in \mathcal{B}_M} Z_b^* .
\label{eq:Wk}
\end{equation}

Since all $X^n_b$ are integer, $Z_{\mathcal{B}_M}^*$ is also an integer.
We can now establish a minimum for step size $W_{\mathcal{B}_M}$ above which identical merging for all blocks $b \in \mathcal{B}_M$ is achievable:

\begin{fact}{\textbf{Minimum Step Size for Identical Merging}}\label{fact:min-step}: a step size $W_{\mathcal{B}_M} > Z_{\mathcal{B}_M}^* $, is large enough for \texttt{floor} function $f(X^n_b)$ in (\ref{eq:floor}) to merge any $X^{n}_b$ in $\mathcal{B}_M$ to a same value $\bar{X}_b$.
\end{fact}

Since each $\mathbf{S}^n$ is a coarse approximation of (and thus is similar to) desired target $\mathbf{T}$, the $\mathbf{S}^n$'s themselves are similar. Hence, the largest difference $Z_b^*$ should be small in the typical case. Indeed, we observe empirically that $Z^*_b$ follows an exponential distribution (one-sided because $Z_b^*$ is non-negative). Fig.\;\ref{fig:Z} shows $Z^*_b$ probability distribution for $k=16$ and $k=32$. We can see that 80\% of the blocks have $Z^*_b \leq 5$.
Assuming that $Z^*_b$ follows a Laplacian distribution, the maximum
$Z_{\mathcal{B}_M}^*$ is typically much larger than the average
$Z^*_b$. This will be shown to be useful for the optimized merging of
Section~\ref{sec:Solving}. 


\begin{figure}[htb]
\begin{minipage}[b]{0.45\linewidth}
  \centering
  \centerline{\includegraphics[width=45mm]{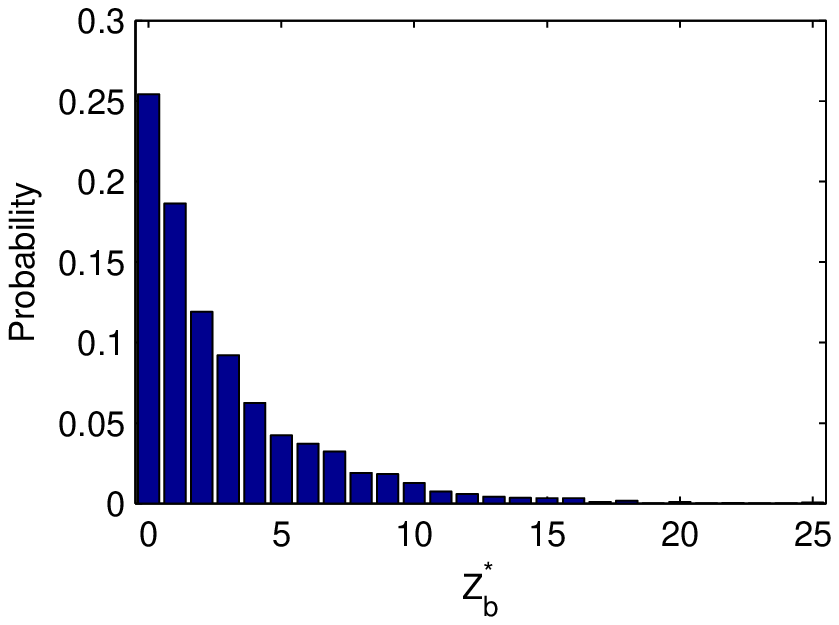}}\medskip
\end{minipage}
\hfill
\begin{minipage}[b]{0.45\linewidth}
  \centering
  \centerline{\includegraphics[width=45mm]{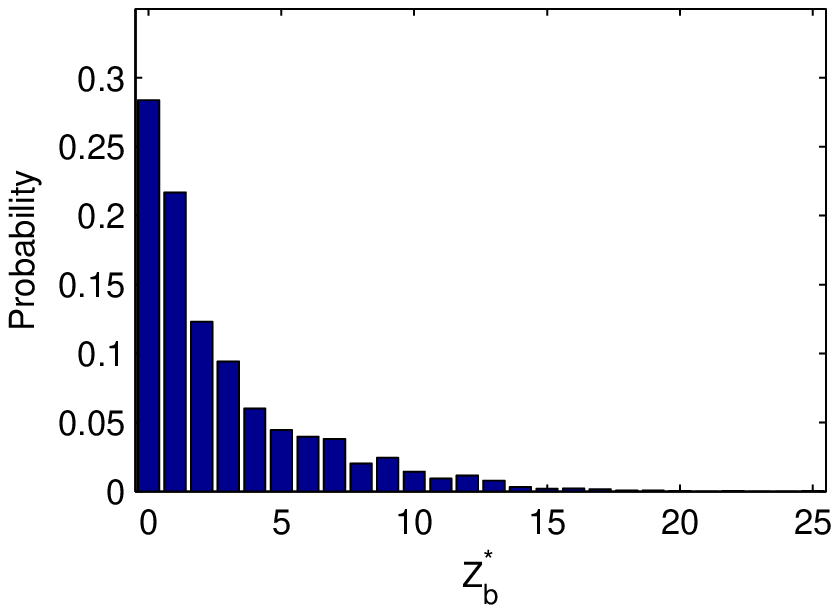}}\medskip
\end{minipage}
\vspace{-0.2in}
\caption{Two examples of probability distribution of $Z^*_b$ with three SI frames at $Q = 1$ for \texttt{Balloons} at frequency $k=16$ and $k=32$.}
\label{fig:Z}
\end{figure}

Fact~\ref{fact:min-step} states that step size $W_{\mathcal{B}_M}$ is wide enough so that $X^{0}_b,\ldots, X^{N}_b$ can all fall on the same 
interval in $f(x)$, as shown in Fig.\;\ref{fig:pwc}. However, given $W_{\mathcal{B}_M}$, shift $c_b$ must still be appropriately chosen \textit{per block} to achieve identical merging. 


Mathematically, identical merging means that the \texttt{floor} function with parameters $c_b$ and $W_{\mathcal{B}_M}$ produces the same integer output for all inputs $X^n_b$, that is: 
\begin{equation}
\left\lfloor \frac{X^{n}_b + c_b}{W_{\mathcal{B}_M}} \right\rfloor = \left\lfloor \frac{X^{0}_b + c_b}{W_{\mathcal{B}_M}} \right\rfloor , ~~~ \forall n \in \{1, \ldots, N\} .
\label{eq:cCondition}
\end{equation}

Thus for all $X^n_b$, we must have 
for some $m \in \mathbb{Z}$ that: 
\begin{equation}\label{equ:W}
    mW_{\mathcal{B}_M} \leq X^{n}_b + c_b < (m+1)W_{\mathcal{B}_M} ,
~~ \forall n \in \{0, \ldots, N\}
\end{equation}
Instead of considering all $X^n_b$'s, it is sufficient to consider only the maximum and minimum values, so that the maximum range for $c_b$ that guarantees identical reconstruction is: 
\begin{equation}
mW_{\mathcal{B}_M} - X^{\min}_b \leq c_b < (m+1)W_{\mathcal{B}_M} - X^{\max}_b 
\label{eq:c_bound}
\end{equation}
for some integer $m$. Note that given step size $W_{\mathcal{B}_M}$, $c_b$ and $c_b + m W_{\mathcal{B}_M}$ lead to the same output: 
\begin{align}
f(x) = &
\left\lfloor \frac{x + c_b + m W_{\mathcal{B}_M}}{W_{\mathcal{B}_M}}\right\rfloor
W_{\mathcal{B}_M} + \frac{W_{\mathcal{B}_M}}{2} - (c_b + m W_{\mathcal{B}_M}) 
\nonumber \\
= & \left\lfloor \frac{x + c_b}{W_{\mathcal{B}_M}}\right\rfloor
W_{\mathcal{B}_M} + \frac{W_{\mathcal{B}_M}}{2} - c_b
\nonumber
\end{align}
Thus it will be sufficient to consider at most $W$ different values of $c_b$ as possible candidates. 

Define $\alpha = X^{\min}_b \bmod W_{\mathcal{B}_M} $ and 
$\beta = X^{\max}_b \bmod W_{\mathcal{B}_M} $ and 
consider the two possible cases. 
\begin{itemize}
\item In case (i) $X^{\min}_b = m W_{\mathcal{B}_M} + \alpha$ and 
$X^{\max}_b = m W_{\mathcal{B}_M} + \beta$, where $\alpha < \beta$, so
that $X^{\min}_b$ and $X^{\max}_b$ fall in the same interval when
there is no shift, $c_b=0$. Hence we can have $ -\alpha \leq c_b < W_{\mathcal{B}_M} - \beta$ in order to keep both $X^{\min}_b$ and $X^{\max}_b$ in the interval $[m W_{\mathcal{B}_M}, (m+1) W_{\mathcal{B}_M})$.
\item In case (ii) $X^{\min}_b = m W_{\mathcal{B}_M} + \alpha$ 
and 
$X^{\max}_b = (m+1) W_{\mathcal{B}_M} + \beta$, where $\beta <
\alpha$, \textit{i.e.}, when $c_b=0$, $X^{\min}_b$ and $X^{\max}_b$ fall in
neighboring intervals. Here we can have $ -\alpha \leq c_b < -\beta $ to move $X^{\max}_b$ down to the interval $[m W_{\mathcal{B}_M}, (m+1) W_{\mathcal{B}_M})$, or have $ W_{\mathcal{B}_M} -\alpha < c_b \leq W_{\mathcal{B}_M} -\beta $ to move $X^{\min}_b$ up to the interval $[(m+1) W_{\mathcal{B}_M}, (m+2) W_{\mathcal{B}_M})$.
\end{itemize} 
Note that the selection of $W_{\mathcal{B}_M}$ (Fact 1) implies that $X^{\max}_b - X^{\min}_b < W_{\mathcal{B}_M}$, and $\alpha = \beta$ only if $X^{\min}_b = X^{\max}_b$, in which case there is no merging needed and any $c_b$ would suffice.
\begin{figure}[htb]
\begin{minipage}[b]{0.47\linewidth}
  \centering
  \centerline{\includegraphics[width=45mm]{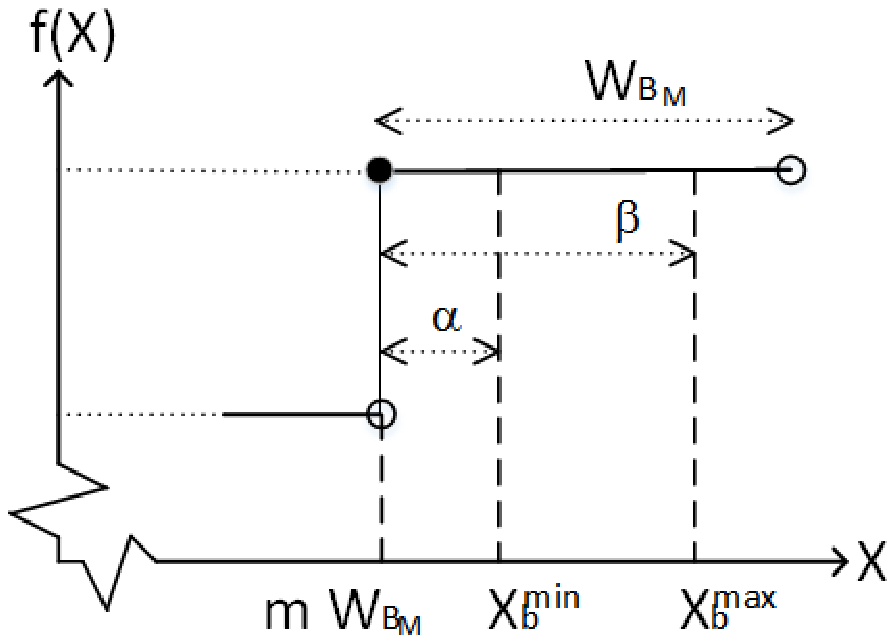}}\medskip
\end{minipage}
\hfill
\begin{minipage}[b]{0.47\linewidth}
  \centering
  \centerline{\includegraphics[width=45mm]{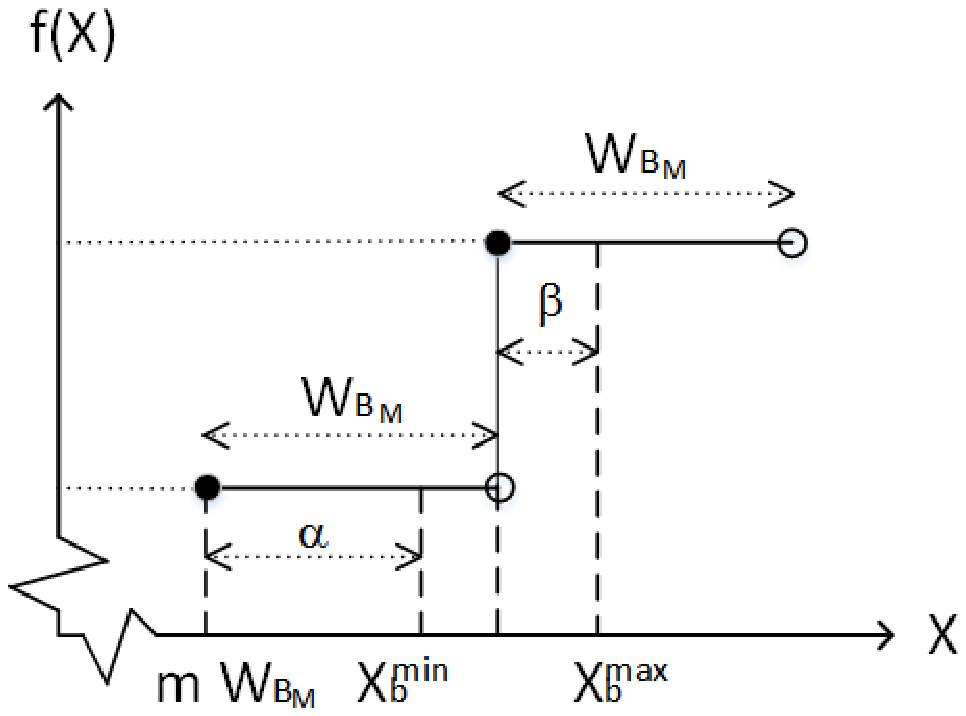}}\medskip
\end{minipage}
\vspace{-0.2in}
\caption{Two cases of $X^{\min}_b$ and $X^{\max}_b$ (left: $\alpha < \beta$ and right: $\alpha > \beta$) and their implications on the feasible range of shift $c_b$.}
\label{fig:alphaBeta}
\end{figure}

The two cases ($\alpha < \beta$ and $\alpha > \beta$) are illustrated in Fig.\;\ref{fig:alphaBeta}. 
Note that given $ X^{\max}_b \geq  X^{\min}_b$ by definition, we will be in Case (ii) whenever $\beta < \alpha$. Thus we can summarize this result as: 
\begin{fact}{\textbf{Maximum Feasible Range $\mathcal{F}_b$ for Shift $c_b$}}\label{fact:feasible}: For the shift $c_b$ to provide identical merging of q-coeffs $X^0_b, \ldots X^N_b$ to a same value $\bar{X}_b$, given step size $W_{\mathcal{B}_M}$
\[
c_b \in \mathcal{F}_b =  [-\alpha, W_{\mathcal{B}_M} - \beta) \;\; {\rm if} \;\; \alpha < \beta 
\]
and
\[
c_b \in \mathcal{F}_b = [W_{\mathcal{B}_M} -\alpha, W_{\mathcal{B}_M} - \beta) \;\; {\rm if} \;\; \alpha > \beta \\
\]
with   
$\alpha = X^{\min}_b \bmod W_{\mathcal{B}_M} $ and 
$\beta = X^{\max}_b \bmod W_{\mathcal{B}_M} $.
\end{fact}

\subsection{Formulation of Merge Frame RD-Optimization}
\label{sec:RD-opt}


In order to formulate the PWC function parameter optimization problem,
we first define distortion, $d_b$, as the squared difference between coefficient $Y_b^0$ of the desired target $\mathbf{T}$ and reconstructed coefficient $f(X^0_b) \, Q$:
\begin{equation}
d_b ~=~ | \, Y_b^0 - f(X^0_b) \, Q \, |^2 .
\label{eq:distort}
\end{equation}
Because shift $c_b$ will be always chosen within the feasible range
defined in Fact~\ref{fact:feasible}, all q-coeffs $X^n_b$ will map to
the same value $f(X^n_b), \forall n \in \{0, \ldots, N\}$. Thus we
only need to compute the distortion for $f(X^0_b)$ in
(\ref{eq:distort}).


For the $k$-th q-coeff in block group $\mathcal{B}_M$, the encoder
will have to transmit to the decoder:
\begin{enumerate}
\item one step size $W_{\mathcal{B}_M}(k) > Z_{\mathcal{B}_M}(k)$ for
  each group $\mathcal{B}_M$.
\item one shift $c_b(k)$ for each block $b$ in group $\mathcal{B}_M$.
\end{enumerate}
The cost of encoding a single $W_{\mathcal{B}_M}(k)$ for all $k$-th
q-coeffs in group $\mathcal{B}_M$ is small, while the cost of encoding
$|\mathcal{B}_M|$ shifts $c_b(k)$ for each of the $k$-th q-coeffs can
be significant. Thus we consider only the rate associated to $c_b(k)$
in our optimization. 

Note that since the high-frequency DCT coefficients of a given code
block are very likely zero, we can insert an \textit{End of Block}
(EOB) flag $E_b$ to signal the remaining high-frequency q-coeffs in
block $b$ in a raster-scan order are 0. Effective use of $E_b$ can
reduce the amount of transmitted PWC function parameters\footnote{In
  the fixed target merging case, $E_b$ is inserted when the remaining
  high-frequency q-coeffs of a block $b$ in target $\mathbf{T}$ are
  exactly zero. In the optimized target case, $E_b$ can be inserted in
  an RD-optimal manner on a per-block basis, similar to what is done
  in coding standards such as H.264~\cite{wiegand03}.}. In summary, we
can define the RD optimized target merging problem as:  
\begin{equation}
\label{equ:opt}
\min_{W_{\mathcal{B}_M}(k), \, c_b(k)}
\quad \sum_{b \in \mathcal{B}_M} D_b + \lambda R_b ,
~~~ \begin{array}{l}
W_{\mathcal{B}_M}(k) > Z_{\mathcal{B}_M}(k) \\
c_b(k) \in \mathcal{F}_b(k) 
\end{array}
\end{equation}
with distortion  $D_b$  and rate $R_b$ for block $b$ calculated as:
\begin{eqnarray*}
  D_b & = & \sum_{k=0}^{E_b} d_b(k) + \sum_{k = E_b+1}^{K-1} Y_b^0(k)^2 \\
  R_b & = & \sum_{k=0}^{E_b} R(c_b(k)) ,
\end{eqnarray*}
where $d_b(k)$ is defined in (\ref{eq:distort}) and $R(c_b(k))$ is the
rate to encode $c_b(k)$. We discuss how we tackle this optimization in
Section~\ref{sec:Solving}.


\section{Fixed Target Merging}
\label{sec:target}

In certain applications, such as the static view switching scenario discussed in Section~\ref{sec:system} and illustrated in Fig.\;\ref{fig:SVS}, the picture interactivity graph is cyclic, so that we may have that  $\Pi_p \rightarrow \Pi_q$ and $\Pi_q \rightarrow \Pi_p$. Because of this interdependency, one cannot directly define a simple target merging optimization, since optimizing the reconstruction for $\Pi_q$ would require first fixing a representation (frame) for $\Pi_p$, but optimizing  $\Pi_p$ would in turn require first fixing a representation for $\Pi_q$. As a simple alternative we propose \textit{fixed target merging}, where the reconstruction target $\mathbf{T}$ for each picture is chosen independently from the SI frames. For example, $\mathbf{T}$ can be the I-frame of the target picture for a given QP.


\subsection{Fixed Target Reconstruction using Merge Operator}


We first show that given a target reconstruction value $a$ and a step size $W$, we can always find a shift $c$ so that $f(x)$ in (\ref{eq:floor}) is such that $f(x) = a$ for all inputs $x$ in the interval $[a - W/2, a+W/2)$. To see this, first write target reconstruction value $a = a_1 W + a_2$, where $a_1$ and $a_2 = a \bmod W$ are integers and  $0 \leq a_2 < W$. Similarly, we write input $x = a_1 W + x_2$ where integer $x_2$ can be bounded:
\begin{align}
a - \frac{W}{2} \leq & ~~x ~ < a + \frac{W}{2} \nonumber \\
a_1 W + a_2 - \frac{W}{2} \leq & ~~ a_1 W + x_2 ~ < a_1 W + a_2 + \frac{W}{2} 
\nonumber \\
a_2 - \frac{W}{2} \leq & ~~ x_2 ~ < a_2 + \frac{W}{2}
\label{eq:boundX2}
\end{align}

We now set $c = \frac{W}{2} - a_2$. We show that this ensures $f(x) = a$ for $x \in [a - W/2, a + W/2)$:
\begin{align}
f(x) & = \left\lfloor \frac{a_1 W + x_2 + \frac{W}{2} - a_2}{W} \right\rfloor W + \frac{W}{2} - \left( \frac{W}{2} - a_2 \right) 
\label{eq:targetMerge} \\
& = a_1 W + a_2  = a . \nonumber
\end{align}
where the second line is true because $x_2 + \frac{W}{2} - a_2$ in the numerator of the ``round-down" operator argument can be bounded  in $[0, W)$ using (\ref{eq:boundX2}):
\begin{align}
a_2 - \frac{W}{2} + \frac{W}{2} - a_2 \leq & ~~ x_2 + \frac{W}{2} - a_2 ~ < a_2 + \frac{W}{2} + \frac{W}{2} - a_2 \nonumber \\
0 \leq & ~~ x_2 + \frac{W}{2} - a_2 ~ < W
\end{align}



Next, recall from Section~\ref{subsec:PWC} that we include the desired target $\mathbf{T}$ as the first SI frame $\mathbf{S}^0$. For a given frequency of a particular block $b$, we first compute the \textit{maximum target difference} $Z_b$ as the largest absolute difference between target q-coeff $X^0_b$ and $X^n_b$ of any SI frame $\mathbf{S}^n$, \textit{i.e.},
\begin{equation}
Z_b = \max_{n \in \{1, \ldots, N\}} 
\left| X^0_b - X^n_b \right| 
\label{eq:maxDiff}
\end{equation}
Based on this we can choose step size and shift based on the following lemma. 
\begin{lemma}
\label{lemma:merge}
Choosing step size $W^{\#}_b = 2 Z_b + 2$ and shift $c_b = W^{\#}_b/2 - X^0_{b,2}$, where $X^0_{b,2} = X^0_b \bmod W^{\#}_b$, guarantees that $f(X^n_b) = X^0_b, ~ \forall n \in \{0, \ldots, N\}$. Note that $W^{\#}_b$ is an even number, and $c$ is an integer as required.  
\end{lemma}
\begin{proof}
Given shift $c_b = W^{\#}_b/2 - X^0_{b,2}$, showing $X^n_b \in [X^0_{b} - W^{\#}_b/2, X^0_{b} + W^{\#}_b/2)$ implies $f(X^n_b) = X^0_b, ~ \forall n \in \{0, \ldots, N\}$. Defining step size $W^{\#}_b = 2 Z_b + 2$ means the required  interval for $X^n_b$ can be rewritten as $[X^0_{b} - Z_b - 1, X^0_{b} + Z_b + 1)$. By the definition of $Z_b$, we know $X^0_b - Z_b \leq X^n_b \leq X^0_b + Z_b$. Hence the required interval for $X^n_b$ is met.
\end{proof}





Note that we can achieve fixed target merging for a given $X^0_b$ as long as the step size is larger than $W^{\#}_b$. For example, we can assign the same step size $W^{\#}_{\mathcal{B}_M}$ for all blocks in a group $\mathcal{B}_M$, so that we reduce the rate overhead: 
\begin{equation}
W^{\#}_{\mathcal{B}_M} = 2 + 2 Z_{\mathcal{B}_M}
\label{eq:perfectStepSize}
\end{equation}
where $Z_{\mathcal{B}_M} = \max_{b \in \mathcal{B}_M} Z_b$ is the \textit{group-wise maximum target difference}, and $Z_b$, the block-wise maximum target difference for block $b$, is computed using (\ref{eq:maxDiff}).
In summary:  
\begin{enumerate}
\item We define a set of blocks $\mathcal{B}_M$ and use $W^{\#}_{\mathcal{B}_M}(k)$ computed using (\ref{eq:perfectStepSize}) for frequency $k$ of all blocks in $\mathcal{B}_M$.
\item For block $b$, we set shift $c_b(k) = W^{\#}_{\mathcal{B}_M}(k)/2 - X^0_{b,2}(k)$, where $X^0_{b,2}(k) = X^0_b(k) \bmod  W^{\#}_{\mathcal{B}_M}(k)$. A different shift is used for each frequency $k$ and block $b$, and transmitted as part of the M-frame along with $W^{\#}_{\mathcal{B}_M}(k)$. 
\end{enumerate}

\section{Optimized Target Merging}
\label{sec:Solving}

We now propose a merging approach based on selecting $W_{\mathcal{B}_M}(k)$ and $c_b(k)$ so as to find a solution to the optimization problem described in Section~\ref{sec:RD-opt}, where we allow the reconstructed value to be different from $X^0_b(k)$. 

If $W_{\mathcal{B}_M}$ is chosen large enough, \textit{i.e.} $W_{\mathcal{B}_M} \geq 2 + 2 Z_{\mathcal{B}_M}$, then we have shown (Lemma~\ref{lemma:merge}) that one can select shift $c_b$ to reconstruct target q-coeff $X^0_b$ exactly. 
However, the shifts are a function of $X^0_{b,2} = X^0_b \mod W_{\mathcal{B}_M}$ (Lemma~\ref{lemma:merge}), and thus we can expect them to have a uniform distribution, which would mean that a rate of the order of $\log_2(W_{\mathcal{B}_M})$ would be required as overhead. In order to reduce this rate, we use two approaches: i) we allow $W_{\mathcal{B}_M}$ to be smaller than required by Lemma~\ref{lemma:merge}, and ii) when multiple choices of $c_b$ provide identical reconstruction, we optimize this choice based on the criteria introduced in Section~\ref{sec:RD-opt}.


\subsection{Selection of  $W_{\mathcal{B}_M}$}
\label{subsec:W}



Note, by definition of  $Z^*_{\mathcal{B}_M}$, we are guaranteed that all $X^n_b$ can be within an interval of size $W_{\mathcal{B}_M} $ as long as $W_{\mathcal{B}_M} > Z^*_{\mathcal{B}_M}$, provided we transmit an appropriate $c_b$ (Fact \ref{fact:min-step}). Reducing $W_{\mathcal{B}_M}$ from $2 + 2 Z_{\mathcal{B}_M}$ can reduce the rate required to transmit $c_b$, since $c_b$ can take at most $W_{\mathcal{B}_M}$ different values. 



As shown in Section~\ref{subsec:PWC} we observe empirically that $Z^*_b$ follows a Laplacian distribution (Fig.\;\ref{fig:Z}). Thus, for a large block group $\mathcal{B}_M$, $Z^*_{\mathcal{B}_M} = \max_{b\in\mathcal{B}_M} Z^*_b$ will be in general much larger than $Z^*_b$. Since $Z^*_b \geq Z_b$, in practice for many blocks $b$ it is thus possible to reconstruct target $X^0_b$ since $W^*_{\mathcal{B}_M} > 2 Z_b + 2$. Thus, we propose to select $W_{\mathcal{B}_M} = Z^*_{\mathcal{B}_M} + 1$, which guarantees that for the worst case block all SI values are in the same interval, with appropriate choice of $c_b$ to be discussed next. 


\subsection{RD-optimal Selection of Shifts}
\label{subsec:RDoptShift}


Given a chosen $W_{\mathcal{B}_M}$, according to Fact~\ref{fact:feasible} there will be multiple values of $c_b$ that guarantee identical reconstruction for all $X^n_b$.  To enable efficient entropy coding of $c_b$, it is desirable to have a skewed probability distribution $P(c_b)$ of $c_b$. 
We design an algorithm to promote a skewed $P(c_b)$ iteratively. We first propose how to initialize $P(c_b)$, and then discuss how to update $P(c_b)$ in subsequent iterations.


We optimize shift $c_b$ via the following RD cost function:
\begin{equation}
\min_{0 \leq c_b < W_{\mathcal{B}_M} \,|\, c_b \in \mathcal{F}_b} d_b + \lambda (- \log P(c_b)) ,
\label{eq:RDopt}
\end{equation}
where the rate term is approximated as the negative log of the probability $P(c_b)$ of candidate $c_b$, and $d_b$ is the distortion term computed using (\ref{eq:distort}). 
%
The difficulty in using objective (\ref{eq:RDopt}) to compute optimal $c^*_b$ lies in how to define $P(c_b)$ \textit{prior} to selection of $c_b$. Our strategy is to initialize a skewed distribution $P(c_b)$ to promote a low coding rate, perform optimization (\ref{eq:RDopt}) for each block $b \in \mathcal{B}_M$, then update $P(c_b)$ based on statistics of the selected $c_b$'s, and repeat until $P(c_b)$ converges.



In order to choose an initial distribution $P(c_b)$, we note 
that a distribution with a small number of spikes has lower entropy than a smooth  distribution (see Fig.~\ref{fig:ProbC} as an example). 
Choosing $c_b$ values following such a discrete distribution (\textit{e.g.}, left in Fig.~\ref{fig:ProbC}) means that we reduce the number of possible $c_b$, which may increase $d_b$. Thus, if $\lambda$ in (\ref{eq:RDopt}) is small, in order to reduce distortion one can increase the number of spikes in $P(c_b)$. In this paper, we propose to induce a multi-spike probability $P(c_b)$, where the appropriate number of spikes depends on the desired tradeoff between distortion and rate in (\ref{eq:RDopt}).

\begin{figure}[htb]
\begin{minipage}[b]{0.45\linewidth}
  \centering
  \centerline{\includegraphics[width=45mm]{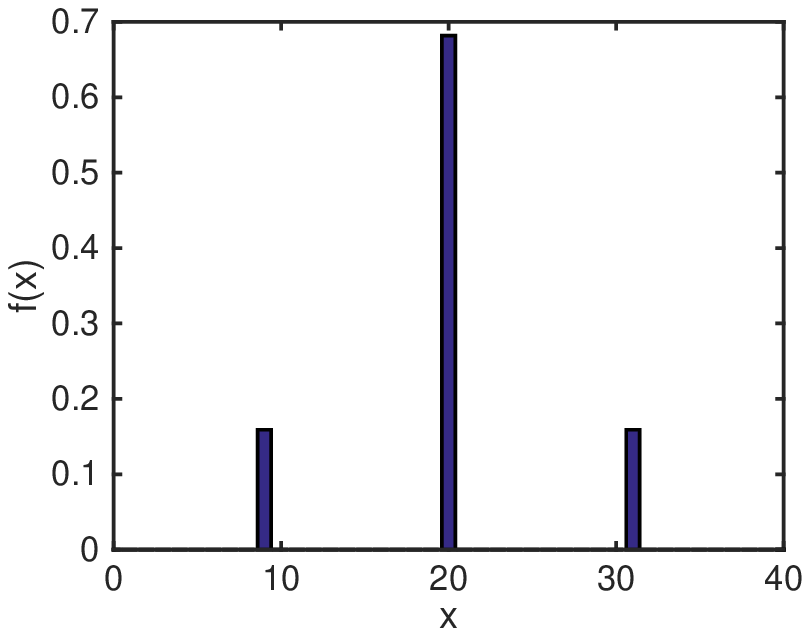}}\medskip
\end{minipage}
\hfill
\begin{minipage}[b]{0.45\linewidth}
  \centering
  \centerline{\includegraphics[width=45mm]{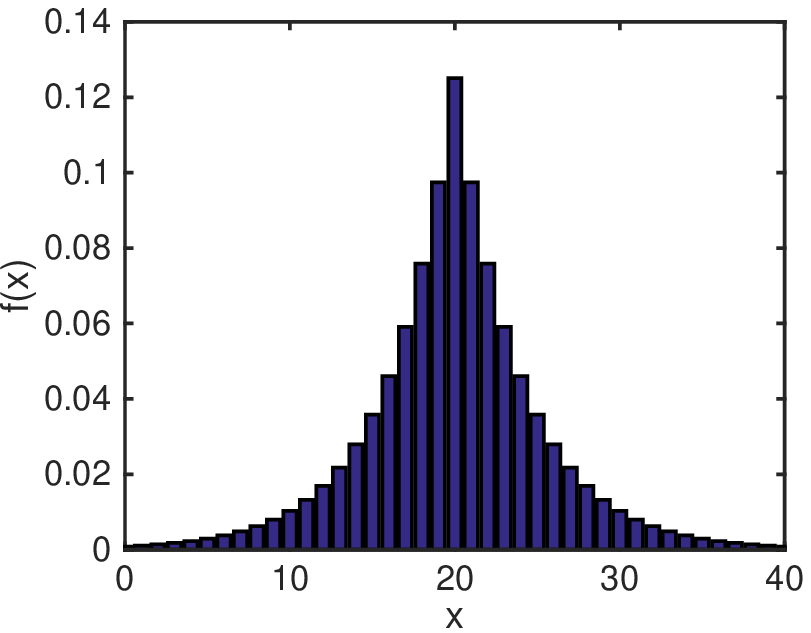}}\medskip
\end{minipage}
\vspace{-0.2in}
\caption{Two examples of shift distribution $P(c_b)$. Left distribution has small number of spikes and has low entropy (1.22). Right distribution is smooth but has high entropy (4.38).}
\label{fig:ProbC}
\end{figure}

Since $c_b$ is constrained to be in the feasible region $\mathcal{F}_b$ defined in Fact~\ref{fact:feasible}, it is possible that when we restrict $c_b$ to just a few values as in Fig.~\ref{fig:ProbC}\,(left), there will be some blocks $b$ for which none of the ``spikes'' in $P(c_b)$ fall within their $\mathcal{F}_b$. In order to guarantee identical reconstruction they must be able to select non-spike values as shifts $c_b$. Thus we propose a ``spike + uniform'' distribution $P(c_b)$:
\begin{equation}
    P(c_b) = \left\{
    \begin{array}{ll}
        ~~ p^s_i	&	\mbox{if} ~ c_b = c^s_i \\
        ~~ p_c	&	\mbox{o.w.}
    \end{array} \right.
\label{eq:Pcb}
\end{equation}
where $\{c^s_1,\ldots, c^s_H\}$ are the $H$ spikes, each with probability $p^s_i$, and $p_c$ is a small constant for non-spike shift values. $p_c$ is chosen so that $P(c_b)$ sums to $1$.

\subsubsection{Computing distribution $P(c_b)$ for fixed $H$}

We now discuss how we compute $P(c_b)$ for given $H$. Empirically we observe that for a reasonable number of spikes (\textit{e.g.}, $H \geq 3$), the majority of blocks (typically $99\%$ or more) in $\mathcal{B}_M$ have at least one spike in their feasible region $\mathcal{F}_b$. Thus, to simplify our computation we first ignore the feasibility constraint and employ an iterative \textit{rate-constrained Lloyd-Max} algorithm (rc-LM)~\cite{sullivan1996efficient} to identify spike locations.

We illustrate the operations of rc-LM to initialize $H$ spike locations for $H=3$ as follows. Let $c^o_b$ be the shift value that minimizes \textit{only} distortion for block $b$. Let $g(c^o)$ be the probability distribution of distortion-minimizing shift $c^o$ for blocks in $\mathcal{B}_m$, where $0 \leq c^o < W_{\mathcal{B}_M}$. $g(c^o)$ can be computed empirically for group $\mathcal{B}_m$. 
Without loss of generality, we define quantization bins for the three spikes $c^s_1$, $c^s_2$ and $c^s_3$ as $[0, b_1)$, $[b_1, b_2)$ and $[b_2, W_{\mathcal{B}_M})$ respectively. The expected distortion $D(\{c^s_i\})$ given three spikes is:
\begin{equation}
\sum\limits_{c^o = 0}^{b_1 - 1} |c^o-c^s_1|^2 g(c^o) 
+ \sum\limits_{c^o = b_1}^{b_2 - 1} |c^o-c^s_2|^2 g(c^o) 
+ \sum\limits_{c^o = b_2}^{W_{\mathcal{B}_M}-1} |c^o-c^s_2|^2 g(c^o)
\end{equation}
where $D(\{c^s_i\})$ is computed as the sum of squared difference between $c^o$ and spike $c^s_i$ in the bin that $c^o$ is assigned to. Having defined distortion $D(\{c^s_i\})$, the initial spike locations $c^s_i$ given $H$ spikes can be found as follows: i) construct $H$ spikes evenly spaced in the interval $[0, W_{\mathcal{B}_M})$, ii) use conventional Lloyd-Max algorithm with no rate constraints to converge to a set of $H$ bin centroids $c^s_i$.

Next, adding consideration for rate, the RD cost of the three spikes can then be written as:

\vspace{-0.1in}
\begin{small}
\begin{equation}
\begin{array}{l}
D(\{c^s_i\}) 
+ \lambda \left( -\log (\sum\limits_{c^o=0}^{b_1-1} g(c^o)) 
-\log (\sum\limits_{c^o=b_1}^{b_2-1} g(c^o))
-\log (\sum\limits_{c^o=b_2}^{W_{\mathcal{B}_M}-1} g(c^o))  \right)
\end{array}
\label{eq:rdSpike}
\end{equation}
\end{small}\noindent
(\ref{eq:rdSpike}) is essentially the aggregate of RD costs (\ref{eq:RDopt}) for all blocks in $\mathcal{B}_M$. 

To minimize (\ref{eq:rdSpike}), rc-LM alternately optimizes bin boundaries $b_i$ and spike locations $c^s_i$ at a time until convergence. Given spikes $c^s_i$ are fixed, each bin boundary $b_i$ is optimized via exhaustive search in the range $[c^s_i, c^s_{i+1})$ to minimize both rate and distortion in (\ref{eq:rdSpike}). Given bin boundaries $b_i$ are fixed, optimal $c^s_i$ can be computed simply as the bin average:
\begin{equation}
c^s_i = \frac{\sum_{c^o = b_i}^{b_{i+1}-1} g(c^o) c^o}{\sum_{c^o=b_i}^{b_{i+1}-1} g(c^o)}
\label{eq:centroid}
\end{equation}
where $b_0 = 0$ and $b_3 = W_{\mathcal{B}_M}$. 

Upon convergence, we can then identify the small fraction of blocks with no spikes in their feasible regions $\mathcal{F}_b$ and assign an appropriate constant $p_c$ so that $P(c_b)$ is well defined according to (\ref{eq:Pcb}). Computing $P(c_b)$ with $H$ spikes where $H \neq 3$ can be done similarly.




\subsubsection{Finding the optimal $P(c_b)$}
\label{subsubsec:prob}

To find the optimal $P(c_b)$, we add an outer loop for this $P(c_b)$ construction procedure to search for the optimal number of spikes $H$. Pseudo-code of the complete algorithm is shown in Algorithm \ref{alg:ProbC}. We note that in practice, we observe that the number of iterations until convergence is small.

\begin{algorithm}[htb]
	\caption{Computing the optimal shift distribution $P(c_b)$}
	\label{alg:ProbC}
	\begin{algorithmic}[1]
		\FOR{each number of spikes $H \in [1, W_{\mathcal{B}_M}]$}
			\STATE Initialize distribution $P^o(c_b)$ via LM;
			\STATE $t = 0$;
			\REPEAT
				\STATE $t = t + 1$;
				\STATE Update $H$ spike locations $c^s_i$ via (\ref{eq:centroid});
				\STATE Update bin boundaries $b_i$ by minimizing (\ref{eq:rdSpike});
				\STATE Compute $p_c$ for a new $P^t(c_b)$;
			\UNTIL{ $\| P^{t-1}(c_b) -  P^{t}(c_b) \| \leq \epsilon$  }
		\ENDFOR
	\end{algorithmic}
\end{algorithm}





\subsection{Comparison with Coset Coding}
\label{subsec:coset}


We now discuss the similarity between our proposed approaches and coset coding methods in DSC~\cite{pradhan:03}. 
Consider first fixed target merging of one q-coeff of a single block $b$. In a scalar implementation of coset coding, given possible SI values $X^n_b, n \in \{1, \ldots, N\}$, seen as ``noisy'' versions of a target $X^0_b$, the largest difference $Z_{b} = \max_{n} | X^n_b - X^0_b|$ with respect to $X^0_b$ is first computed. The size of the coset $W$ is then selected such that $W > 2 Z_{b}$. The coset index $i_b = X^0_b \bmod W$ is computed at the encoder for transmission. 

At the decoder, the reconstructed value $\hat{X}_b$ is the integer closest to received SI $X^n_b$ with the same coset index $i_b$, \textit{i.e.},
\begin{equation}
\hat{X}_b = \arg \min_{X \in \mathbb{Z}} |X^n_b - X| 
~~~ \mbox{s.t.} ~~ i_b = X \bmod W
\label{eq:coset}
\end{equation}

Using the aforementioned coset coding scheme for blocks $b \in \mathcal{B}_M$, coding of $i_b = X^0_b \bmod W = X^0_{b,2}$ per block is necessary, where coset size $W$ is chosen such that $W > 2 Z_{\mathcal{B}_M}$. In our fixed target merging scheme using PWC functions, we code a shift $c_b = W^{\#}_{\mathcal{B}_M}/2 - X^0_{b,2}$ for each block $b$, where step size $W^{\#}_{\mathcal{B}_M}$ is also proportional to $2 Z_{\mathcal{B}_M}$. Comparing the two schemes one can see that the number of choices that need to be sent to the decoder is the same (one of $W^{\#}_{\mathcal{B}_M}$ possible values in both cases). Both the shift value $c_b$ and $i_b$ are functions of $X^0_{b,2}$, the LSBs of $X^0_b$, which are likely to have an approximately uniform distribution. Thus so the overhead rate should be the same for both coset coding and fixed target merging. 

Consider now the optimized merging case. In this scenario we are able to choose $W_{\mathcal{B}_M} = Z^*_{\mathcal{B}_M} + 1$---likely much smaller than $2 Z_{\mathcal{B}_M} \leq 2 Z^*_{\mathcal{B}_M}$---so that we can still guarantee identical reconstruction, with a reduction in rate that comes at the cost of an increase in distortion. As for the coset coding approach, if we were to reduce to choose a smaller 
$W_{\mathcal{B}_M}$ as well, we in fact can no longer guarantee identical reconstruction. This is because when  
$W_{\mathcal{B}_M} < 2 Z_{\mathcal{B}_M}$ there will be cases where not all the $X^n_b$ are in the same interval, and thus the same $i_b$ will lead to two different values at the decoder depending on the SI received. This imperfect merging will lead to undesirable coding drift in the following predicted frames, as discussed in Section \ref{sec:system}.

\section{Experiments}
\label{sec:results}
We first discuss the general experimental setup and M-frame parameter selection (Section~\ref{subsec:ExperiSet}). We then verify the effectiveness of our proposed ``Spike + Uniform" distribution (Section~\ref{subsec:verify}). Next, we compare the performance of our M-frame in three different situations: 1) static view switching ({\em Scenario 1}  in Section~\ref{subsec:S1}); 2) switching among streams of different rates for the same single-view video ({\em Scenario 2} in Section~\ref{subsec:S2}), and 3) dynamic view switching of multi-view videos of different viewpoints and encoded in the same bit-rate ({\em Scenario 3} in Section~\ref{subsec:S3}). 

\subsection{Experimental Setup}
\label{subsec:ExperiSet}

We use four different multiview video test sequences with resolution 1024x768 for scenarios 1 and 3: \texttt{Balloons}, \texttt{Kendo}\footnote{http://www.tanimoto.nuee.nagoya-u.ac.jp/mpeg/mpeg\_ftv.html}, \texttt{Lovebird1} and \texttt{Newspaper}\footnote{ftp://203.253.128.142}. The viewpoints of each sequence are shown in Table~\ref{tab:SVS}. For scenario 2, we use four single-view video sequences with resolution 1920x1080: \texttt{BasketballDrive}, \texttt{Cactus}, \texttt{Kimono1} and \texttt{ParkScene}\footnote{ftp://ftp.tnt.uni-hannover.de/testsequences/}.

\begin{table}[htb]
\begin{small}
\begin{center}
\renewcommand{\arraystretch}{1.2}
\renewcommand{\multirowsetup}{\centering}
\caption{Viewpoints of each multiview sequences.}\label{tab:SVS}
\begin{tabular}{|c|c|}
\hline
Sequence Name	&   Viewpoints  \\  \hline
Balloons		&	1, 3, 5     \\  \hline
Kendo			&	1, 3, 5     \\  \hline
Lovebird1		&	4, 6, 8     \\  \hline
Newspaper		&	3, 4, 5     \\  \hline
\end{tabular}
\end{center}
\end{small}
\end{table}

We compare the coding performance of our proposed scheme against two schemes\footnote{Here  $QP_{A}$ denotes the quantization parameter for coding DCT coefficients in approach $A$}: SP-frame~\cite{karczewicz03} in H.264 and D-frame proposed in~\cite{cheung2010rate}. $QP$ for D-frame is set to be equal to $QP_{SI}$ to maintain consistent quality. For multi-view scenarios 1 and 3, we encoded three streams from three viewpoints: the center view was set as the target, to which the other two side views can switch at a defined switching point. For Scenario 2, we encoded the single-view video in three different bit-rates and then switched among them. The bit-rates for the three streams were decided according to \textit{additive increase multiple decrease} (AIMD) rate control behavior in TCP and TFRC~\cite{chiu1989analysis}: one stream has twice the target stream's bit-rate, while the other has slightly smaller bit-rate (0.9 times of the target stream's bit-rate). The results are shown in plots of PSNR versus coding rate for a switched frame.

M-frame parameters are selected as follows. In Scenario 1, different $QP_{M}$ will result in different rates, and so we set $QP_{M}$ to equal to $QP_{SI}$, as was done for D-frames. However, for optimized target merging, coding rate is determined mainly by the number of spikes in the distribution, and not $QP_{M}$. 
In our experiments, as similarly done in High Efficiency Video Coding (HEVC), we first empirically compute $\lambda$ as a function of the SI frame's $QP_{SI}$:
\begin{equation}
	\lambda = 2^{0.6QP_{SI}-12} \text{,}
\end{equation}
The number of spikes in the distribution is driven by the selected $\lambda$. We then set $QP_{M} = 1$ to maintain small quantization error.
For mode selection among \textit{skip}, \textit{intra} and \textit{merge}, for each block $b$ we first examine q-coeffs $X^n_b(k)$ of $N$ SI frames. If $X^n_b(k)$ of all $K$ frequencies are identical across the SI frames, then block $b$ is coded as \textit{skip}. Otherwise, selection between \textit{intra} and \textit{merge} is done based on a RD criteria.


In HEVC, large code block sizes are introduced which bring significant coding gain on high resolution sequences~\cite{sullivan2012overview}. Motivated by this observation, we also investigated the effect of different block sizes ($4 \times 4$, $8 \times 8$, $16 \times 16$) on coding performance. We also compare our current proposal against the performance of our previous work~\cite{dai13}, where block size is fixed at $8\times 8$, initial probability distribution of shift $P(c_b)$ is not optimized, and no RD-optimized EOB flag is employed. The corresponding PSNR-bitrate curves for scenario 3 are shown in Fig.\;\ref{fig:ESB}.

\begin{figure}[htb]
	\begin{minipage}[b]{.45\linewidth}
		\centering
		\centerline{\includegraphics[width=4.8cm]{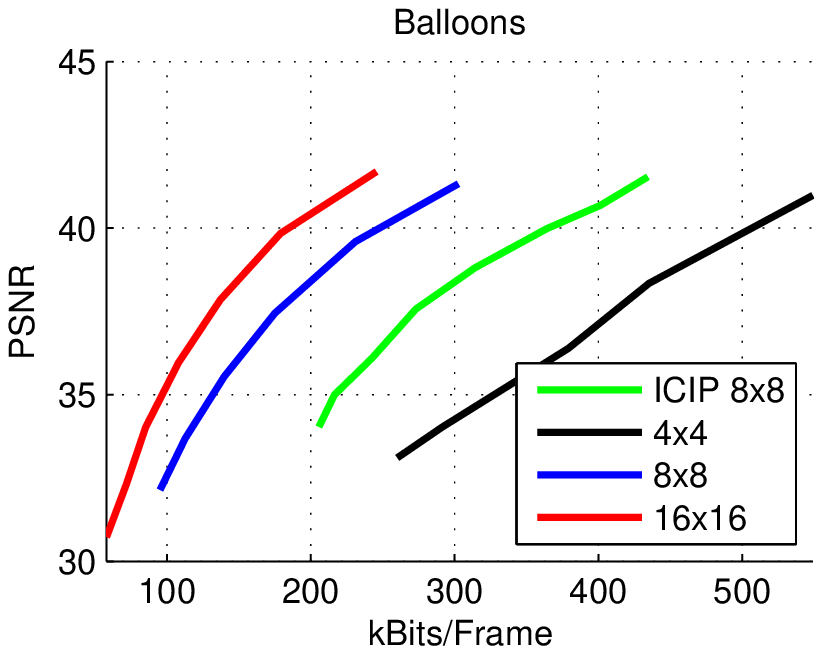}}
		\centerline{(a) \texttt{Balloons}}\medskip
	\end{minipage}
	\hfill
	\begin{minipage}[b]{.45\linewidth}
		\centering
		\centerline{\includegraphics[width=4.8cm]{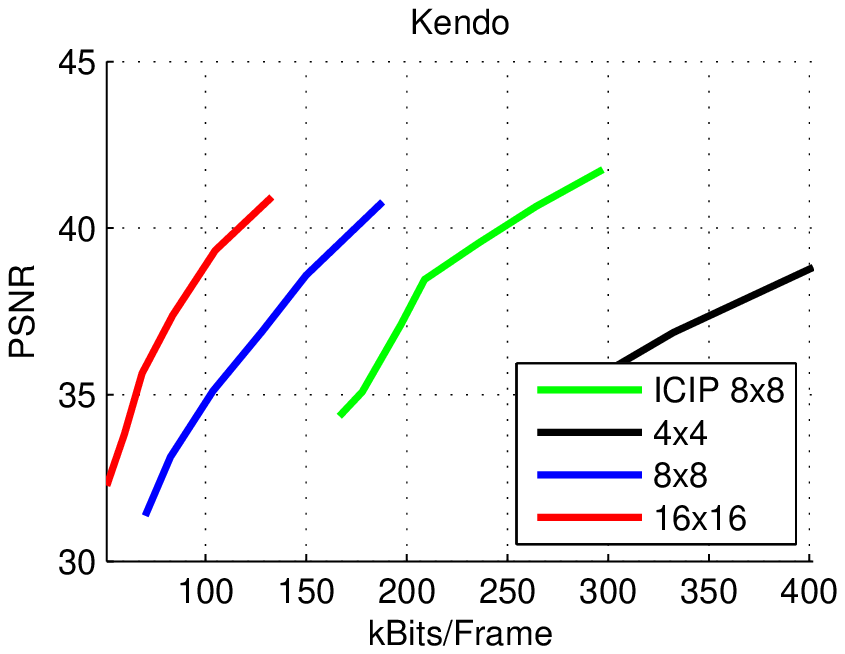}}
		\centerline{(b) \texttt{Kendo}}\medskip
	\end{minipage}
	\begin{minipage}[b]{.45\linewidth}
		\centering
		\centerline{\includegraphics[width=4.8cm]{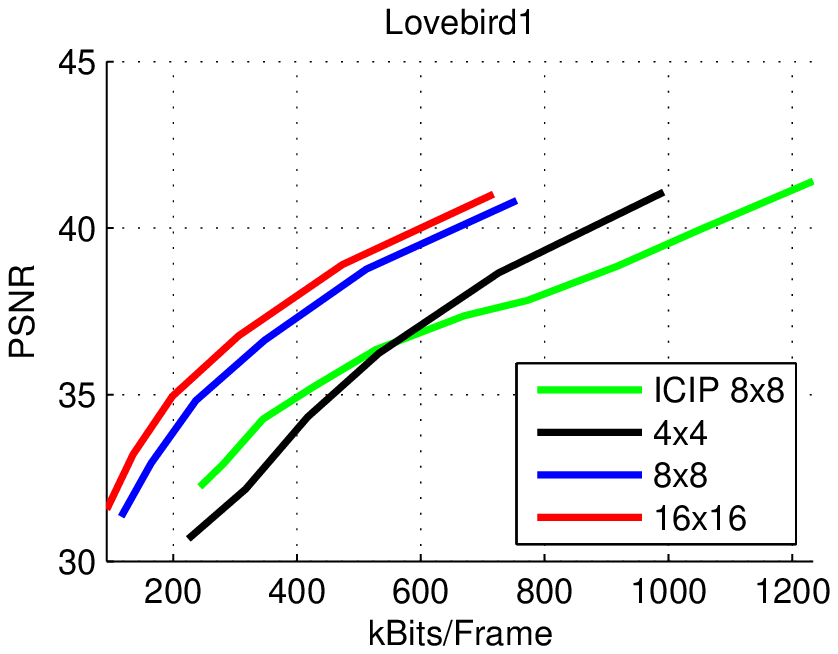}}
		\centerline{(c) \texttt{Lovebird1}}\medskip
	\end{minipage}
	\hfill
	\begin{minipage}[b]{.45\linewidth}
		\centering
		\centerline{\includegraphics[width=4.8cm]{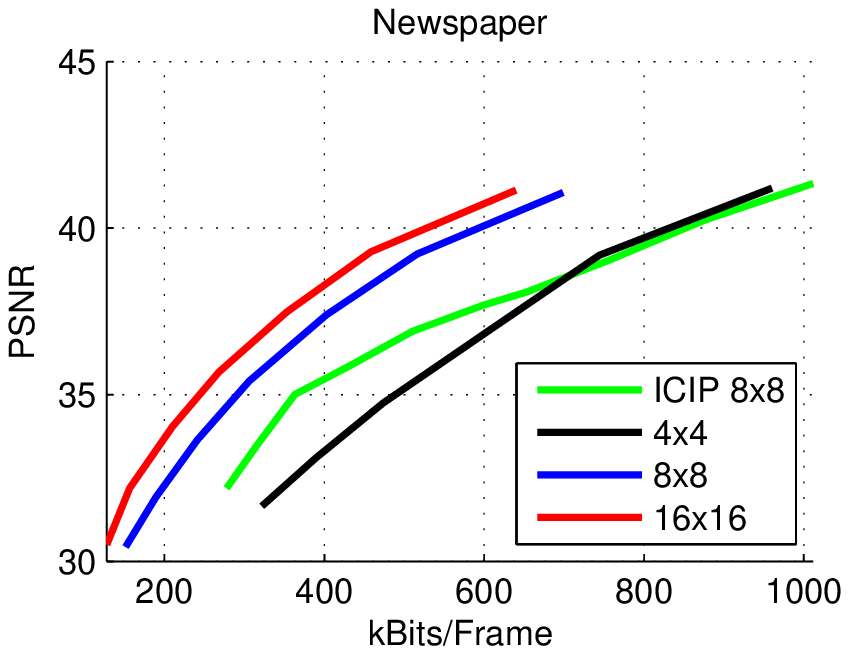}}
		\centerline{(d) \texttt{Newspaper}}\medskip
	\end{minipage}
	\vspace{-0.15in}
	\caption{PSNR v.s. encoding rate comparison with different block sizes for sequences \texttt{Balloons}, \texttt{Kendo}, \texttt{Lovebird1} and \texttt{Newspaper}.}
	\label{fig:ESB}
\end{figure}

From Fig.\;\ref{fig:ESB}, we observe that block size $16 \times 16$ provides the best coding performance at all bit-rates. One reason for the superior performance of large blocks in M-frame is the following: because SI frames are already reconstructions of the target frames (albeit slightly different), motion compensation is not necessary, so the benefit of smaller blocks typical in video coding is diminished.
We note that in general an optimal block size per frame can be selected by the encoder \textit{a priori} and encoded as side information to inform the decoder. In the following experiments, the block size will be fixed at $16 \times 16$ for best performance.

Further, we observe also that our proposed method achieves a significant coding performance gain compared to our previous method in \cite{dai13} over all bit-rate regions, showing the effectiveness of our newly proposed optimization techniques.

\subsection{Effectiveness of ``Spike + Uniform'' Distribution}
\label{subsec:verify}

In order to verify the effectiveness of our proposed ``Spike + Uniform'' (\texttt{SpU}) probability distribution $P(c_b)$ for shift parameter $c_b$, we choose a competing na\"ive distribution for $P(c_b)$ as follows: first, we compute distortion-minimizing $g(c^0)$ as the initial probability distribution. Next, we compute the RD-optimal $c_b$ for each block $b \in \mathcal{B}_M$ via (\ref{eq:RDopt}) for a single iteration using the initialized probability distribution and compute a new $P^\prime(c_b)$. This $P^\prime(c_b)$ is then used to compute the rate to encode each $c_b$ of a merge block $b$. The difference between $P^\prime(c_b)$ and our proposed $P^t(c_b)$ is that $P^\prime(c_b)$ in general is an arbitrarily shaped distribution, not a skewed ``spiky" distribution. Experimental results of M-frame using these distributions are shown in Fig.\;\ref{fig:Naive}.

\begin{figure}[htb]
	\begin{minipage}[b]{.45\linewidth}
		\centering
		\centerline{\includegraphics[width=4.8cm]{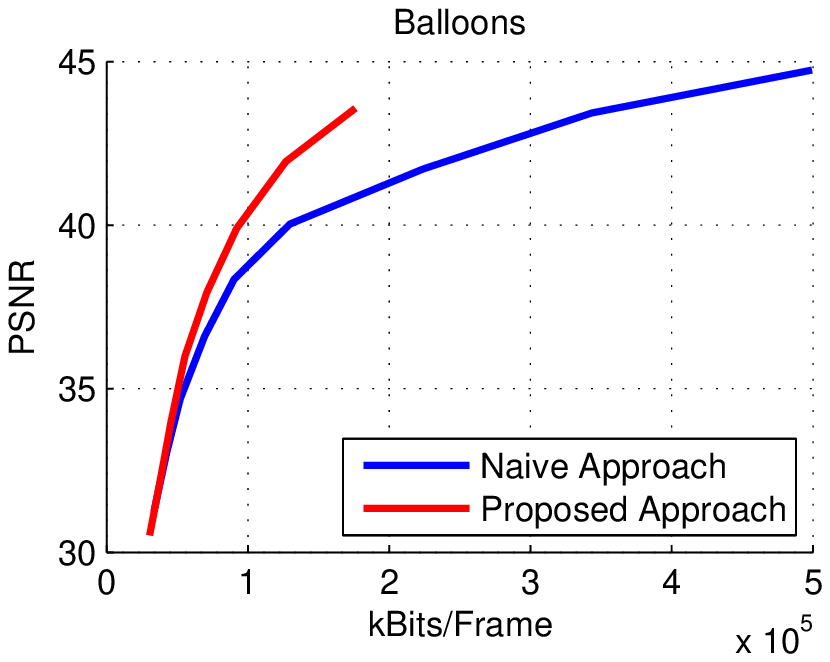}}
		\centerline{(a) \texttt{Balloons}}\medskip
	\end{minipage}
	\hfill
	\begin{minipage}[b]{.45\linewidth}
		\centering
		\centerline{\includegraphics[width=4.8cm]{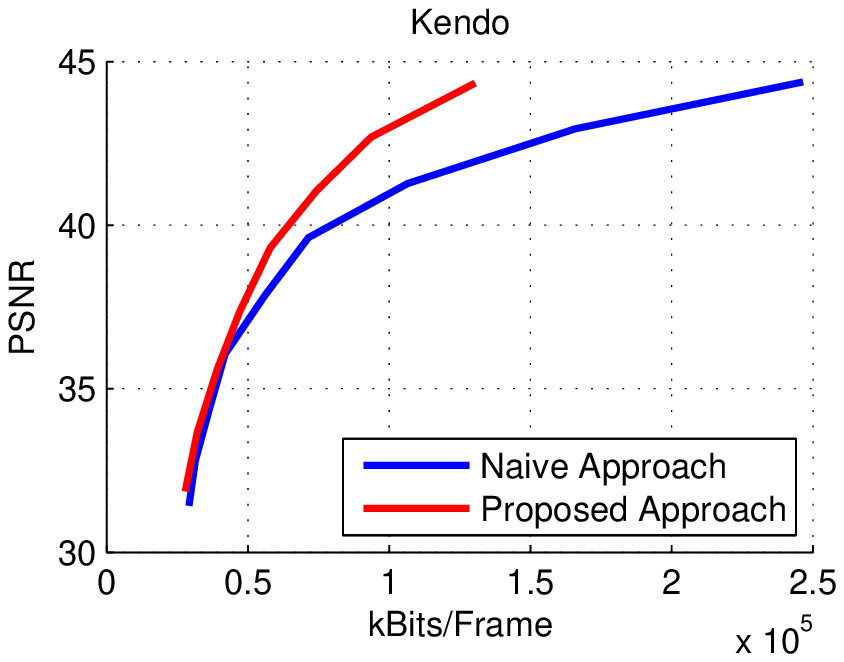}}
		\centerline{(b) \texttt{Kendo}}\medskip
	\end{minipage}
	\vspace{-0.15in}
	\caption{PSNR v.s. encoding rate comparison with different block sizes for sequences \texttt{Balloons}, \texttt{Kendo}.}
	\label{fig:Naive}
\end{figure}

We observe from Fig.\;\ref{fig:Naive} that our proposed \texttt{SpU} distribution outperforms the na\"ive distribution in the high bit-rate region and is comparable in the low bit-rate region. This is because in the low bit-rate region $\lambda$ is very large, so that for any initial distribution, after one iteration, there will only remain one spike, and the number of iterations required for convergence is very small.

\subsection{Scenario 1: Static View Switching}
\label{subsec:S1}

We first test our proposed M-frame in the static view switching scenario for multi-view sequences. Three views are encoded using same $QP$. The fixed target merging algorithm described in Section~\ref{sec:target} is used to facilitate switching to neighboring views among pictures of the same instant, as shown in Fig.~\ref{fig:SVS}. 


Specifically, we constructed M- / D- frames to enable static view-switching from view 1 or 3 to target view 2. We first use H.264 to encode two SI frames (P-frames) using $\Pi_{2, 2}$ as the target and $\Pi_{1, 2}$ and $\Pi_{3, 2}$ as predictors, respectively. This results in encoded rates $\mathcal{R}_{1, 2}$ and $\mathcal{R}_{2, 2}$ for the two SI frames, respectively. Then we encoded a M- / D- frame to merge these two SI frames identically to $\Pi_{2, 2}$. The corresponding rates for M-frame and D-frame are $\mathcal{R}^M_{2,2}$ and $\mathcal{R}^D_{2,2}$, respectively. Since SP-frame in H.264 cannot perform fixed target merging, it is not tested in this scenario.

We assume that the switching probability is equal on both view 1 and 3, which is 0.5. Then the overall rate for the D-frame is calculated as:
\begin{equation}
\mathcal{R}^{D} = \frac{\mathcal{R}_{1, 2} + \mathcal{R}_{3, 2}}{2} + \mathcal{R}^D_{2, 2} \text{.}
\end{equation}

Also, the overall rate for our proposed M-frame using fixed target merging scheme is calculated as:
\begin{equation}
\mathcal{R}^{M} = \frac{\mathcal{R}_{1, 2} + \mathcal{R}_{3, 2}}{2} + \mathcal{R}^M_{2, 2} \text{.}
\end{equation}

\begin{figure}[htb]
\begin{minipage}[b]{.45\linewidth}
  \centering
  \centerline{\includegraphics[width=4.8cm]{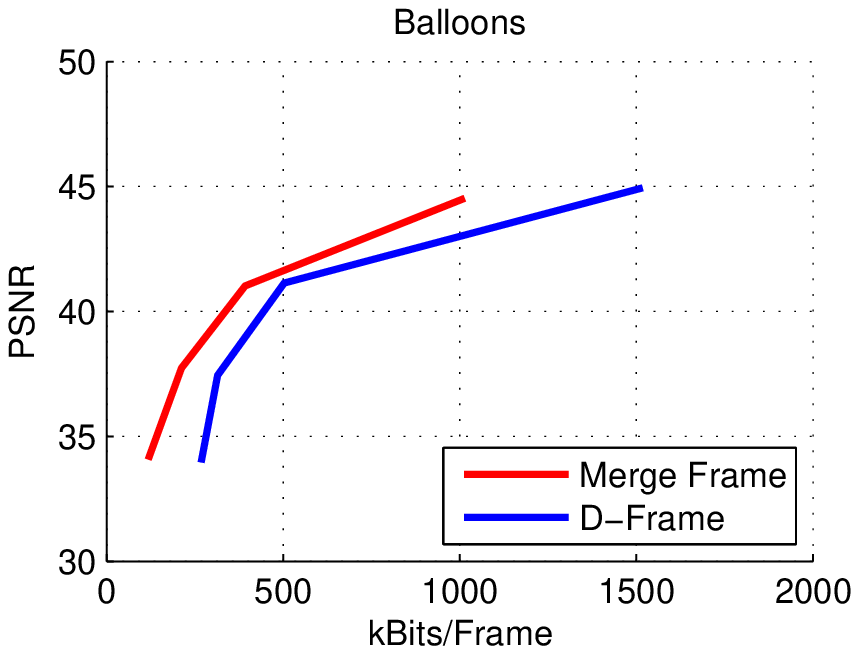}}
  \centerline{(a) \texttt{Balloons}}\medskip
\end{minipage}
\hfill
\begin{minipage}[b]{.45\linewidth}
  \centering
  \centerline{\includegraphics[width=4.8cm]{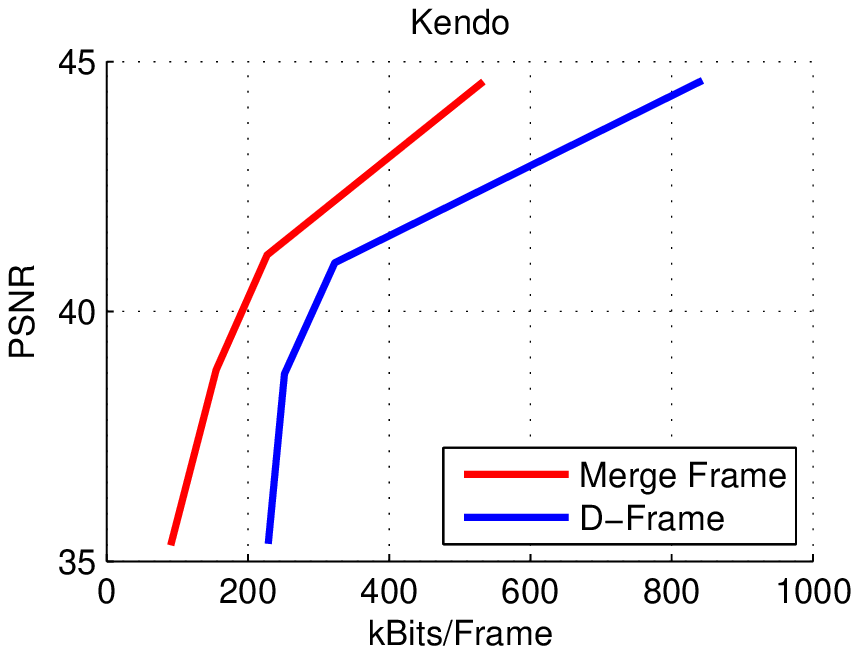}}
  \centerline{(b) \texttt{Kendo}}\medskip
\end{minipage}
\begin{minipage}[b]{.45\linewidth}
  \centering
  \centerline{\includegraphics[width=4.8cm]{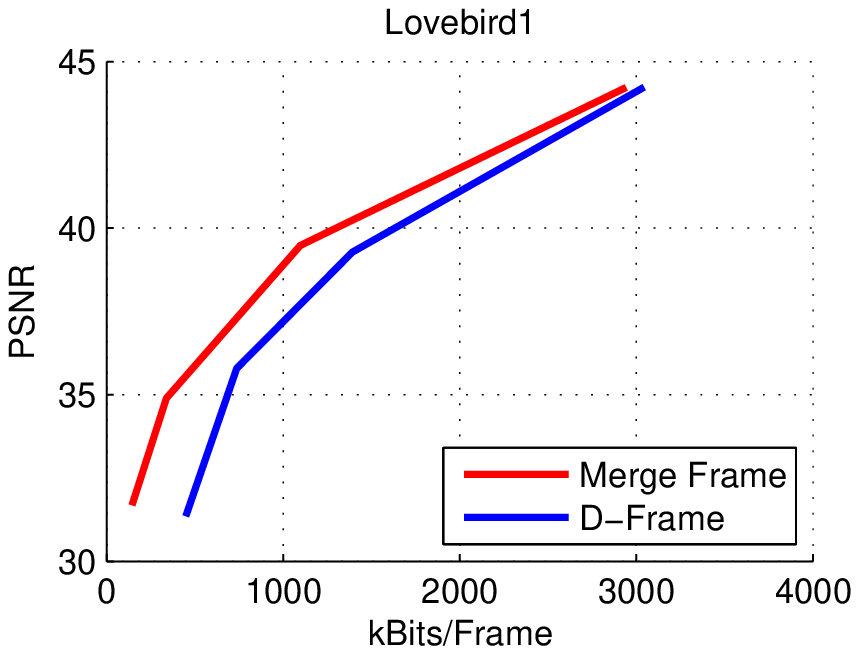}}
  \centerline{(c) \texttt{Lovebird1}}\medskip
\end{minipage}
\hfill
\begin{minipage}[b]{.45\linewidth}
  \centering
  \centerline{\includegraphics[width=4.8cm]{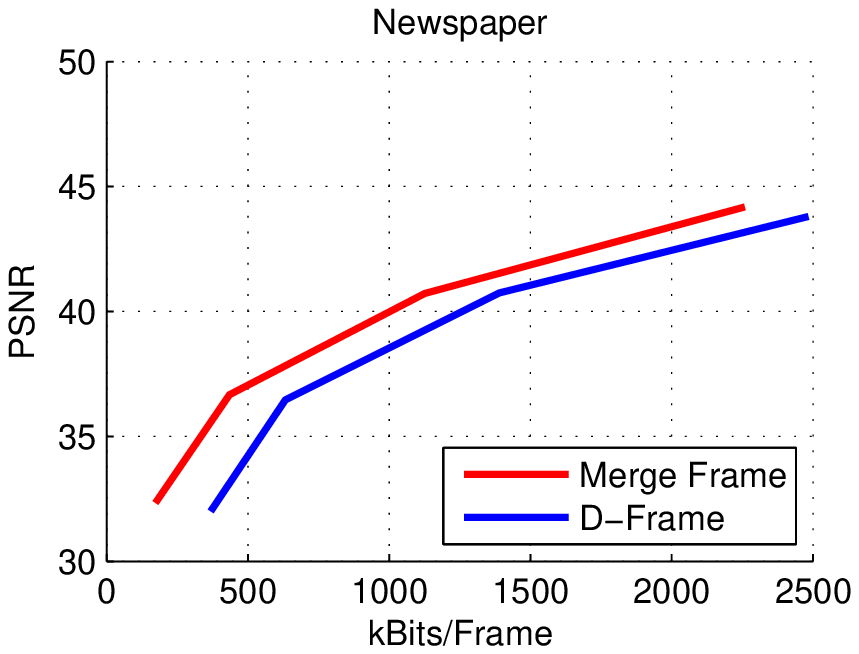}}
  \centerline{(d) \texttt{Newspaper}}\medskip
\end{minipage}
\vspace{-0.15in}
\caption{PSNR v.s. encoding rate comparing proposed M-frame using fixed target merging scheme with D-frame for sequences \texttt{Balloons}, \texttt{Kendo}, \texttt{Lovebird1} and \texttt{Newspaper} in static view switching scenario.}
\label{fig:S1}
\end{figure}

\begin{table}[htb]
	\begin{center}
		\renewcommand{\arraystretch}{1.2}
		\renewcommand{\multirowsetup}{\centering}
		\caption{BD-rate reduction of proposed M-frame using fixed target merging scheme compared to D-frame in static view switching scenario.}\label{tab:S1MD}
		\begin{tabular}{|c|c|}
			\hline
			Sequence Name	&   M-frame \emph{vs.} D-frame	\\  \hline
			Balloons		&	-31.7\%					\\	\hline
			Kendo			&	-40.1\%					\\	\hline
			Lovebird1		&	-35.7\%					\\	\hline
			Newspaper		&	-31.1\%					\\	\hline
		\end{tabular}
	\end{center}
\end{table}

The coding results are shown in Fig.~\ref{fig:S1} and BD-rate~\cite{bjontegaard2008improvements} comparison can be found in Table~\ref{tab:S1MD}. We observe from Table~\ref{tab:S1MD} that our proposed M-frame using fixed target merging scheme achieved up to 40.1\% BD-rate reduction compared to D-frame. Further, from Fig.\;\ref{fig:S1} we observe that our M-frame is better than D-frame in all bit-rate regions, especially in low and high bit-rate region, mainly  
due to the skip block and EOB flag tools. 
In high bit-rate region, due to the small distortion in SI frames, more blocks will be classified into skip block, which efficiently reduces the bits to encode the M-frame,  while in low bit-rate region more coefficients are set to zero and skipped due to the EOB flag. This shows the effectiveness of our proposed M-frame using fixed target merging scheme compared to the D-frame. 

\subsection{Scenario 2: Bit-rate Adaptation}
\label{subsec:S2}

We next conducted experiments of bitrate adaptation scenario for single-view video sequences. M-frame is encoded in a RD-optimized manner, described in section~\ref{sec:Solving} with the system framework shown in Fig.~\ref{fig:DVS}. Three streams of different rates are encoded according to AIMD rate control behavior. 


We constructed M- / D- frames to enable stream-switching from stream 1, 2 or 3 to target stream 2 under different bit-rates. We first encode three SI frames using $\Pi_{2, 2}$ as target and $\Pi_{1, 1}$, $\Pi_{2, 1}$ and $\Pi_{3, 1}$ as reference respectively. This results in encoded rate $\mathcal{R}_{1, 1}$, $\mathcal{R}_{2, 1}$ and $\mathcal{R}_{3, 1}$ for the three SI frames, respectively. Then we encoded a M- / D-frame to merge these three SI frames into an identical frame. The corresponding rate for M-frame and D-frame are $\mathcal{R}^M_{2,2}$ and $\mathcal{R}^D_{2,2}$, respectively.


We also constructed SP-frames to enable stream-switching from stream 1, 2 or 3 to target stream 2. We first encoded a primary SP-frame using $\Pi_{2, 2}$ as target and $\Pi_{2, 1}$ as reference. We then losslessly encoded two secondary SP-frames using the primary SP-frame as target and $\Pi_{1, 1}$, $\Pi_{3, 1}$ as reference respectively. $\mathcal{R}^S_{2, 1}$ denotes the rate for primary SP-frame while $\mathcal{R}^S_{1, 1}$ and $\mathcal{R}^S_{3, 1}$ denote the rate for two secondary SP-frames.

As measure for transmission rate, we consider both the average and worst case code rate during a stream-switch. For average case, in the absence of application-dependent information, we assume that the probability of stream-switching is equal for all views. Thus, the overall rate for RD optimized M-frame is calculated as: 
\begin{equation}
	\mathcal{R}^{M}_{T_A} = \frac{\mathcal{R}_{1, 1} + \mathcal{R}_{2, 1} + \mathcal{R}_{3, 1}}{3} + \mathcal{R}^M_{2, 2} \text{.}
\end{equation}

The overall rate for D-frame is calculated as:
\begin{equation}
	\mathcal{R}^{D}_{T_A} = \frac{\mathcal{R}_{1, 1} + \mathcal{R}_{2, 1} + \mathcal{R}_{3, 1}}{3} + \mathcal{R}^D_{2, 2} \text{.}
\end{equation}

The overall rate for SP-frame is calculated as:
\begin{equation}
	\mathcal{R}^{SP}_{T_A} = \frac{\mathcal{R}^S_{1, 1} + \mathcal{R}^S_{2, 1} + \mathcal{R}^S_{3, 1}}{3}	\text{.}
\end{equation}

\begin{figure}[htb]
\begin{minipage}[b]{.45\linewidth}
  \centering
  \centerline{\includegraphics[width=4.8cm]{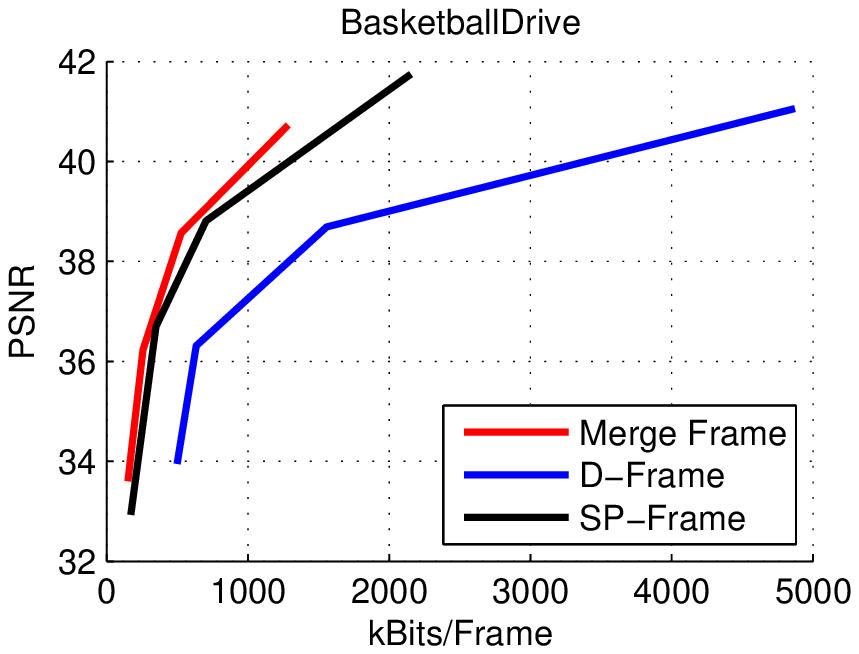}}
  \centerline{(a) \texttt{BasketballDrive}}\medskip
\end{minipage}
\hfill
\begin{minipage}[b]{.45\linewidth}
  \centering
  \centerline{\includegraphics[width=4.8cm]{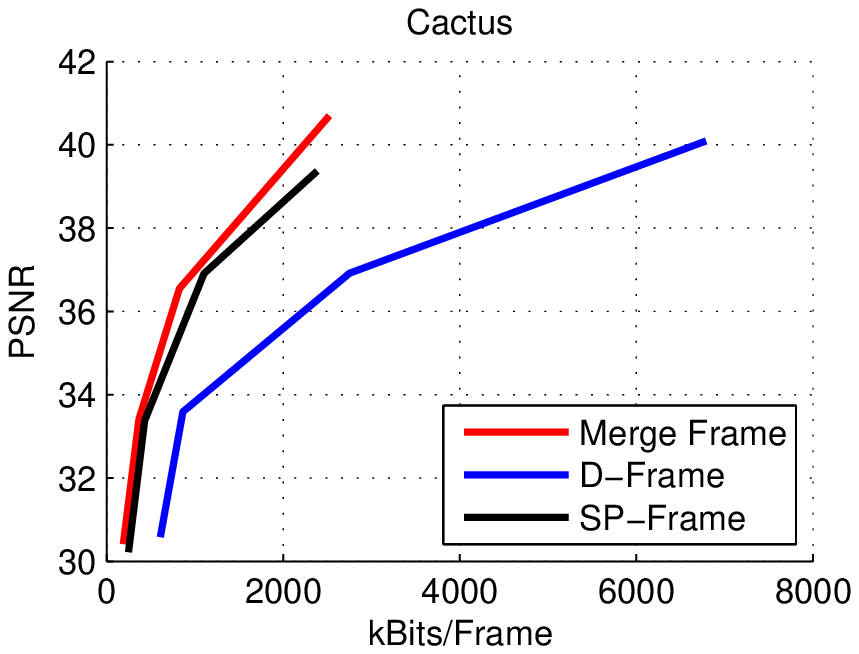}}
  \centerline{(b) \texttt{Cactus}}\medskip
\end{minipage}
\begin{minipage}[b]{.45\linewidth}
  \centering
  \centerline{\includegraphics[width=4.8cm]{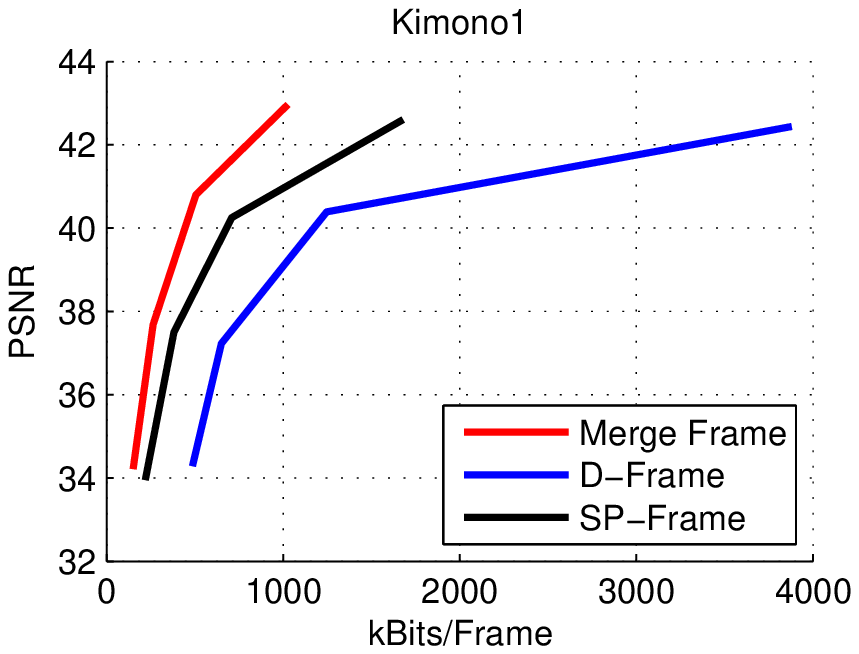}}
  \centerline{(c) \texttt{Kimono1}}\medskip
\end{minipage}
\hfill
\begin{minipage}[b]{.45\linewidth}
  \centering
  \centerline{\includegraphics[width=4.8cm]{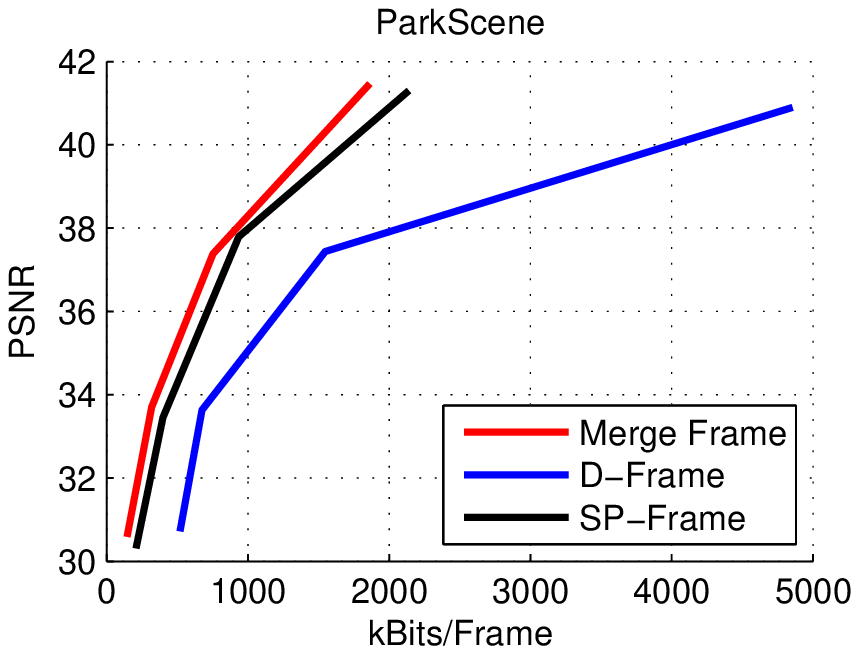}}
  \centerline{(d) \texttt{ParkScene}}\medskip
\end{minipage}
\vspace{-0.15in}
\caption{PSNR versus encoding rate comparing proposed RD-optimized M-frame with D-frame and SP-frame for sequences \texttt{BasketballDrive}, \texttt{Cactus}, \texttt{Kimono1} and \texttt{ParkScene} in average case.}
\label{fig:S4AV}
\end{figure}

\begin{table*}[htb]
	\begin{center}
		\renewcommand{\arraystretch}{1.2}
		\renewcommand{\multirowsetup}{\centering}
		\caption{BD-rate reduction of RD-optimized M-frame compared to D-frame and SP-frame of scenario 2.}\label{tab:S2}
		\begin{tabular}{|c|c|c|c|c|}
			\hline
			\multirow{2}{*}{Sequence Name}	&   \multicolumn{2}{c|}{M-frame \emph{vs.} D-frame}	&	\multicolumn{2}{c|}{M-frame \emph{vs.} SP-frame}	\\  \cline{2-5}
			&	Average Case	&	Worst Case	&	Average Case	&	Worst Case	\\ \hline
			Balloons		&	-63.4\%			&	-63.7\%		&	-17.0\%			&	-39.4\%		\\	\hline
			Kendo			&	-63.5\%			&	-63.2\%		&	-18.8\%			&	-42.1\%		\\	\hline
			Lovebird1		&	-65.6\%			&	-65.4\%		&	-36.3\%			&	-49.9\%		\\	\hline
			Newspaper		&	-56.3\%			&	-56.7\%		&	-19.5\%			&	-43.8\%		\\	\hline
		\end{tabular}
	\end{center}
\end{table*}


The coding results of average case are shown in Fig.~\ref{fig:S4AV} and BD-rate comparison can be found in Table~\ref{tab:S2}. We observe from Table~\ref{tab:S2} that our proposed RD-optimized M-frame achieves up to 65.6\% BD-rate reduction compared to D-frame and 36.3\% BD-rate reduction compared to SP-frame. Moreover, from Fig.\;\ref{fig:S4AV} we observe that our proposed RD-optimized M-frame is better than D-frame and SP-frame in all bit-rate regions. Note that for the SP-frame case, if the switching probability to the primary SP-frame is higher, it will result in a smaller average rate.

For worst case, the code rate for M-frame is calculated as: 
\begin{equation}
\mathcal{R}^{M}_{T_W} = \max(\mathcal{R}_{1, 1}, \mathcal{R}_{2, 1}, \mathcal{R}_{3, 1}) + \mathcal{R}^M_{2, 2} \text{.}
\end{equation}

The rate for D-frame is calculated as:
\begin{equation}
\mathcal{R}^{D}_{T_W} = \max(\mathcal{R}_{1, 1}, \mathcal{R}_{2, 1}, \mathcal{R}_{3, 1}) + \mathcal{R}^D_{2, 2} \text{.}
\end{equation}

The rate for SP-frame is calculated as:
\begin{equation}
\mathcal{R}^{SP}_{T_W} = \max(\mathcal{R}^S_{1, 1}, \mathcal{R}^S_{2, 1}, \mathcal{R}^S_{3, 1})	\text{.}
\end{equation}

\begin{figure}[htb]
\begin{minipage}[b]{.45\linewidth}
  \centering
  \centerline{\includegraphics[width=4.8cm]{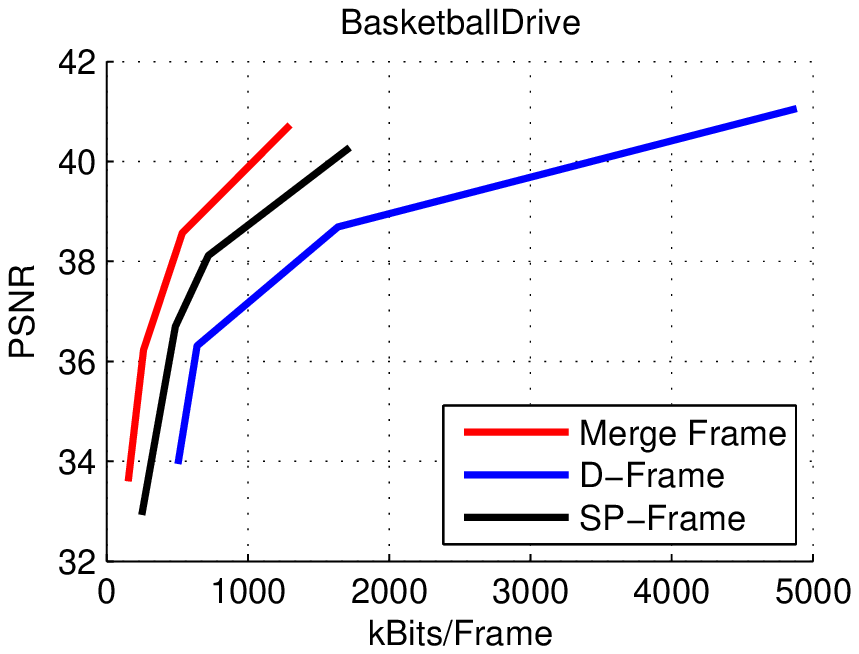}}
  \centerline{(a) \texttt{BasketballDrive}}\medskip
\end{minipage}
\hfill
\begin{minipage}[b]{.45\linewidth}
  \centering
  \centerline{\includegraphics[width=4.8cm]{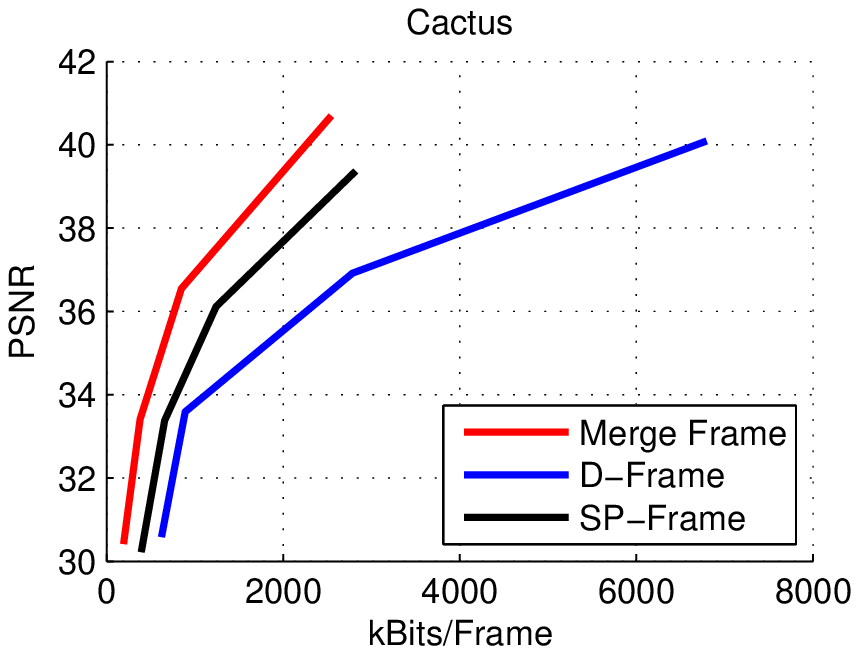}}
  \centerline{(b) \texttt{Cactus}}\medskip
\end{minipage}
\begin{minipage}[b]{.45\linewidth}
  \centering
  \centerline{\includegraphics[width=4.8cm]{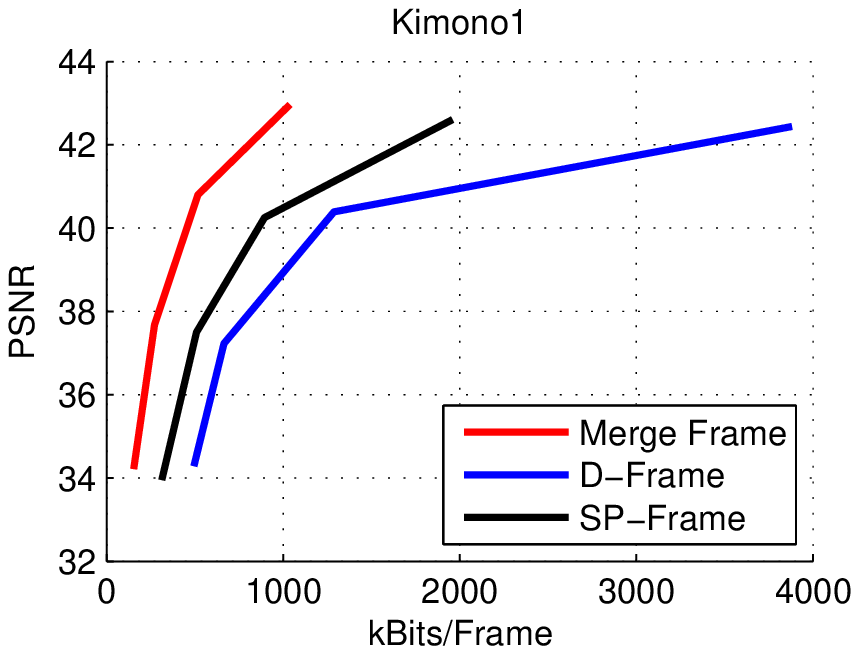}}
  \centerline{(c) \texttt{Kimono1}}\medskip
\end{minipage}
\hfill
\begin{minipage}[b]{.45\linewidth}
  \centering
  \centerline{\includegraphics[width=4.8cm]{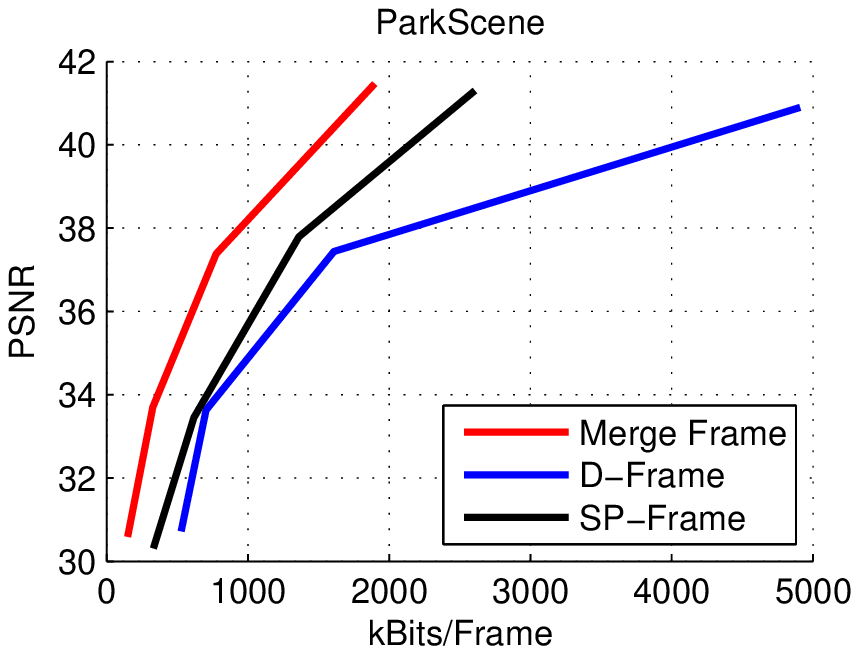}}
  \centerline{(d) \texttt{ParkScene}}\medskip
\end{minipage}
\vspace{-0.15in}
\caption{PSNR versus encoding rate comparing RD-optimized M-frame with D-frame and SP-frame for sequences \texttt{BasketballDrive}, \texttt{Cactus}, \texttt{Kimono1} and \texttt{ParkScene} in worst case.}
\label{fig:S4WS}
\end{figure}


The coding results of worst case are shown in Fig.~\ref{fig:S4WS} and BD-rate comparison can be found in Table~\ref{tab:S2}. We observe from Table~\ref{tab:S2} that our proposed RD-optimized M-frame achieves up to 65.4\% BD-rate reduction compared to D-frame and 49.9\% BD-rate reduction compared to SP-frame.

We observe in Table~\ref{tab:S2} that the performance difference between average and worst case for D-frame is small. However, for SP-frame the performance difference between average and worst case is large. This is due to lossless coding in secondary SP-frames, resulting in a much larger size than primary SP-frame (typically 10 times larger). 


\subsection{Scenario 3: Dynamic View Switching}
\label{subsec:S3}

Finally we conducted experiments of dynamic view switching scenario for multiview video sequences. Three views are encoded using same $QP$. The detailed frame structure for M-frame, D-frame and SP-frame are the same as in Section~\ref{subsec:S2}. Also, the overall rate calculation for average and worst case are identical too.

\begin{figure}[htb]
\begin{minipage}[b]{.45\linewidth}
  \centering
  \centerline{\includegraphics[width=4.8cm]{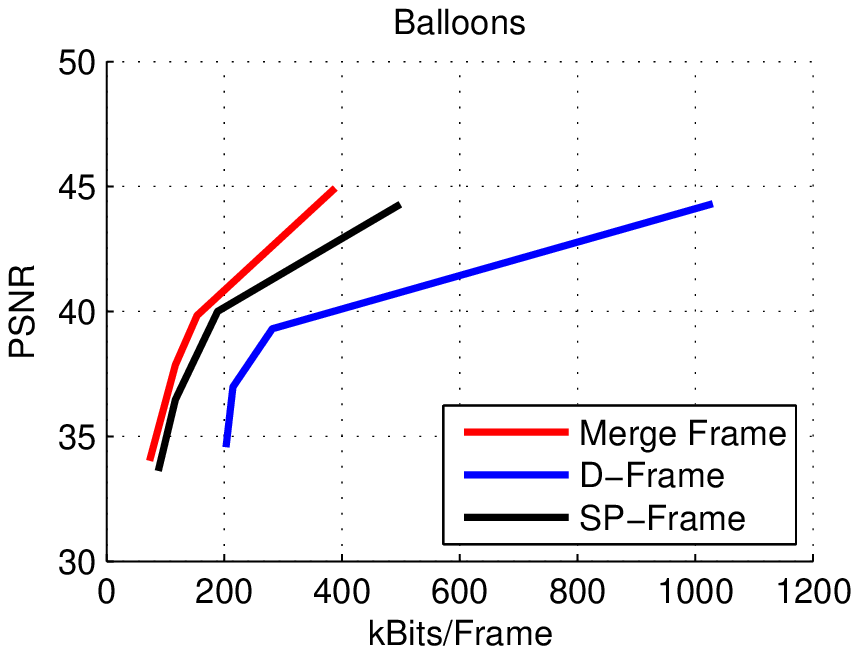}}
  \centerline{(a) \texttt{Balloons}}\medskip
\end{minipage}
\hfill
\begin{minipage}[b]{.45\linewidth}
  \centering
  \centerline{\includegraphics[width=4.8cm]{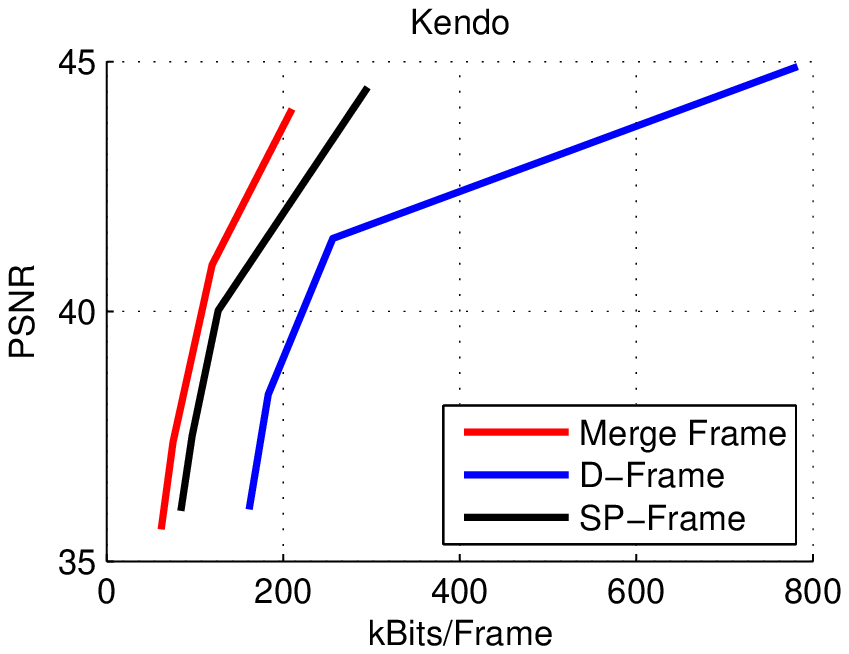}}
  \centerline{(b) \texttt{Kendo}}\medskip
\end{minipage}
\begin{minipage}[b]{.45\linewidth}
  \centering
  \centerline{\includegraphics[width=4.8cm]{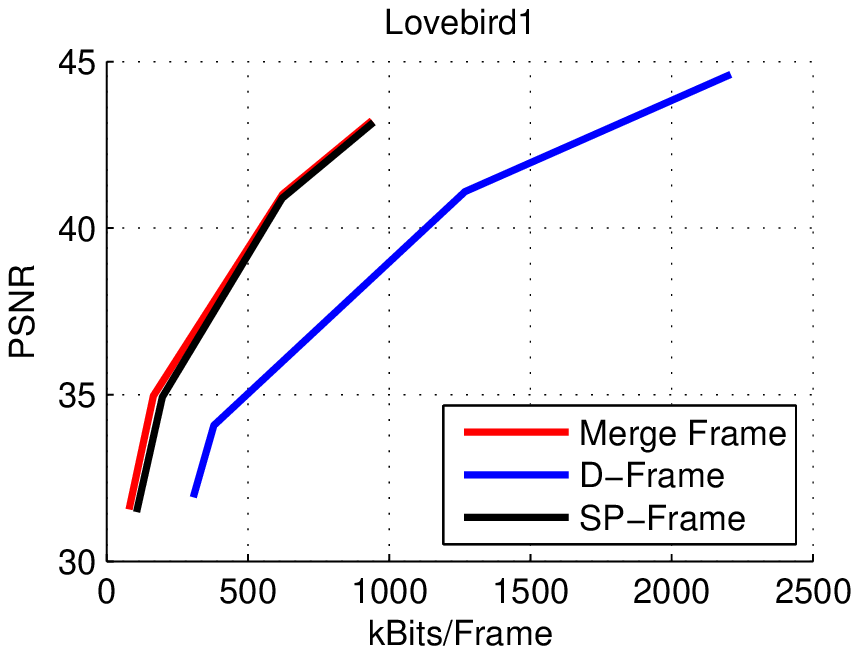}}
  \centerline{(c) \texttt{Lovebird1}}\medskip
\end{minipage}
\hfill
\begin{minipage}[b]{.45\linewidth}
  \centering
  \centerline{\includegraphics[width=4.8cm]{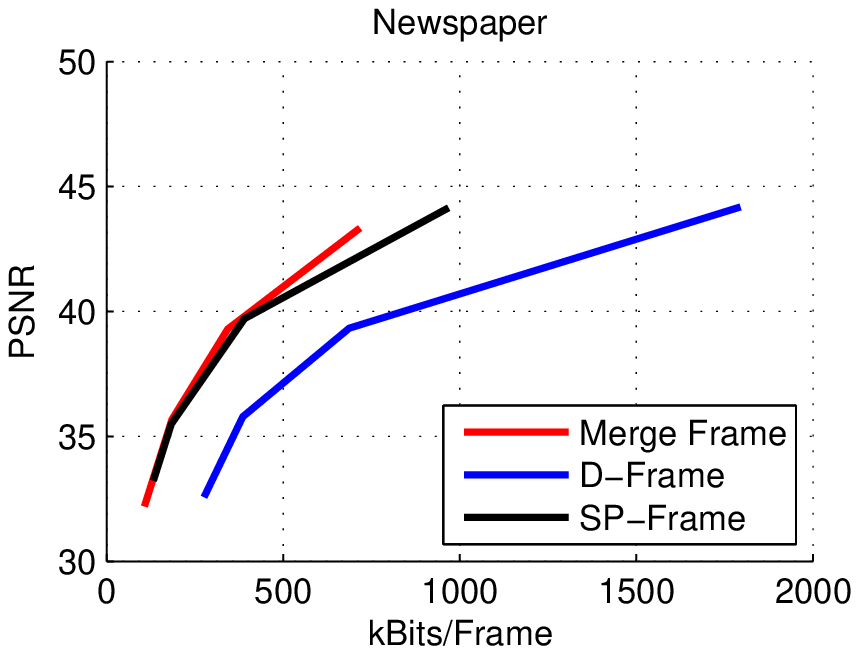}}
  \centerline{(d) \texttt{Newspaper}}\medskip
\end{minipage}
\vspace{-0.15in}
\caption{PSNR versus encoding rate comparing proposed RD-optimized M-frame with D-frame and SP-frame for sequences for sequences \texttt{Balloons}, \texttt{Kendo}, \texttt{Lovebird1} and \texttt{Newspaper} in average case.}
\label{fig:S3AV}
\end{figure}

\begin{table*}[htb]
	\begin{center}
		\renewcommand{\arraystretch}{1.2}
		\renewcommand{\multirowsetup}{\centering}
		\caption{BD-rate reduction of RD-optimized M-frame compared to D-frame and SP-frame of scenario 3.}\label{tab:S3}
		\begin{tabular}{|c|c|c|c|c|}
			\hline
			\multirow{2}{*}{Sequence Name}	&   \multicolumn{2}{c|}{M-frame \emph{vs.} D-frame}	&	\multicolumn{2}{c|}{M-frame \emph{vs.} SP-frame}	\\  \cline{2-5}
							&	Average Case	&	Worst Case	&	Average Case	&	Worst Case	\\ \hline
			Balloons		&	-55.1\%			&	-53.0\%		&	-19.2\%			&	-35.0\%		\\	\hline
			Kendo			&	-53.8\%			&	-53.6\%		&	-19.3\%			&	-36.4\%		\\	\hline
			Lovebird1		&	-57.5\%			&	-58.7\%		&	-11.3\%			&	-28.7\%		\\	\hline
			Newspaper		&	-51.6\%			&	-50.4\%		&	-5.0\%			&	-12.9\%		\\	\hline
		\end{tabular}
	\end{center}
\end{table*}

\begin{figure}[htb]
\begin{minipage}[b]{.45\linewidth}
  \centering
  \centerline{\includegraphics[width=4.8cm]{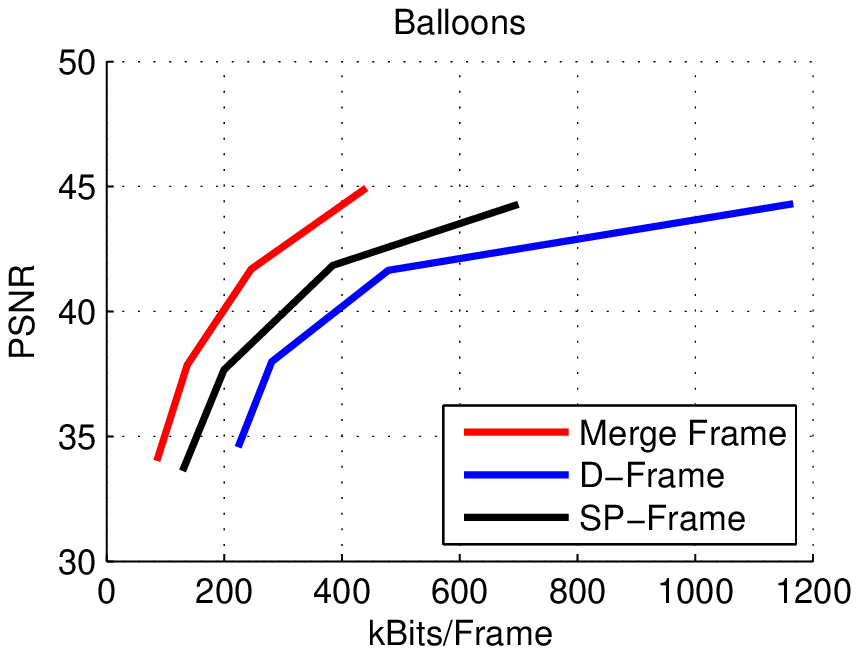}}
  \centerline{(a) \texttt{Balloons}}\medskip
\end{minipage}
\hfill
\begin{minipage}[b]{.45\linewidth}
  \centering
  \centerline{\includegraphics[width=4.8cm]{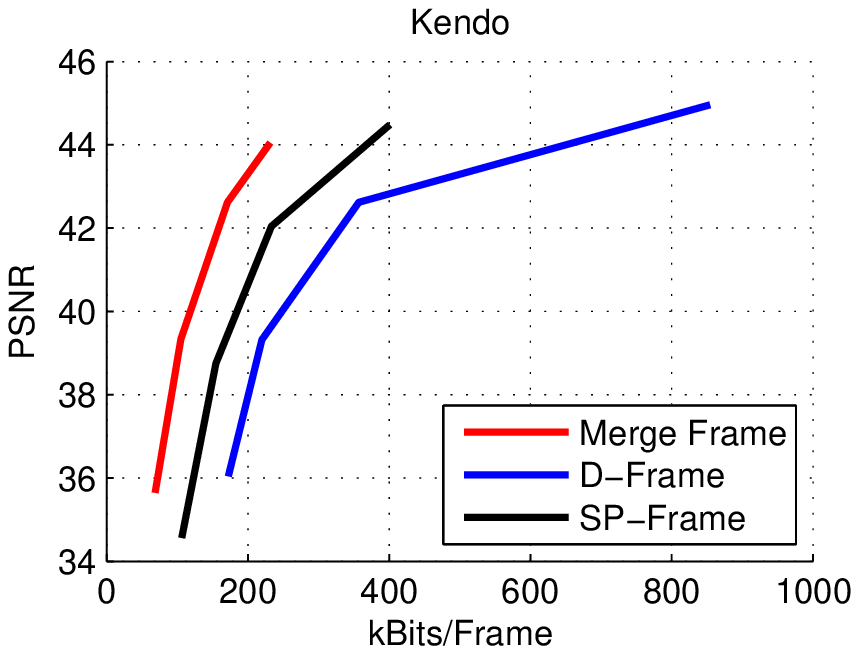}}
  \centerline{(b) \texttt{Kendo}}\medskip
\end{minipage}
\begin{minipage}[b]{.45\linewidth}
  \centering
  \centerline{\includegraphics[width=4.8cm]{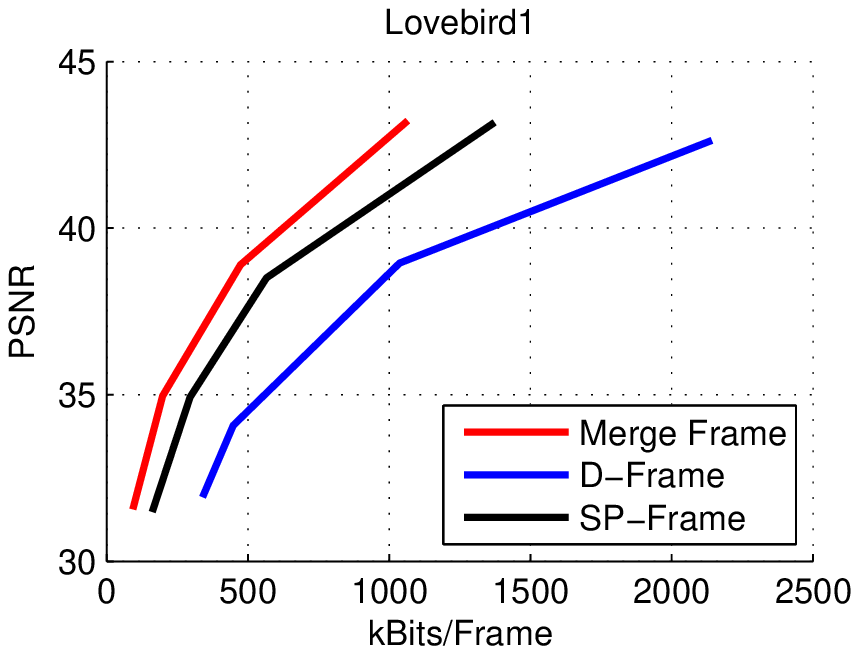}}
  \centerline{(a) \texttt{Lovebird1}}\medskip
\end{minipage}
\hfill
\begin{minipage}[b]{.45\linewidth}
  \centering
  \centerline{\includegraphics[width=4.8cm]{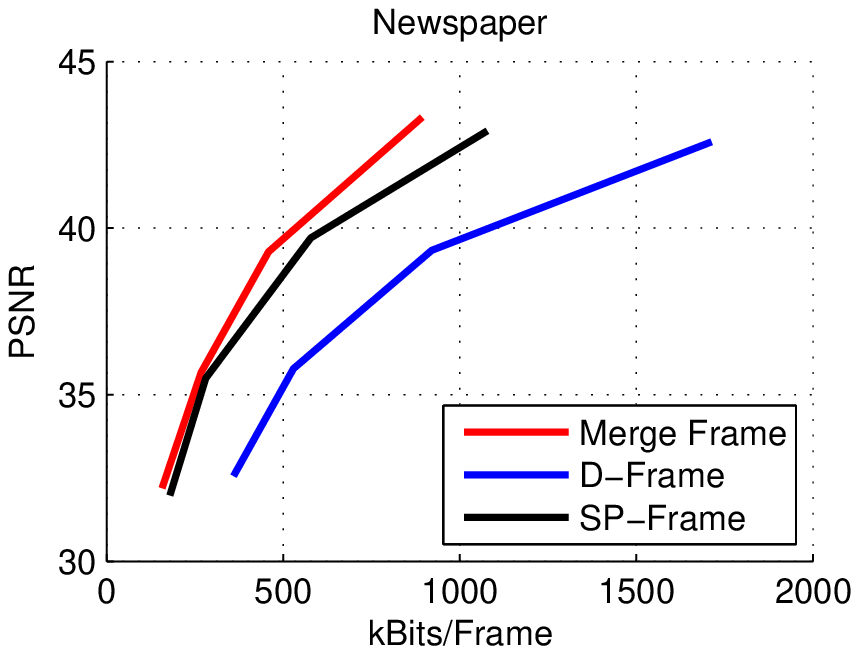}}
  \centerline{(b) \texttt{Newspaper}}\medskip
\end{minipage}
\vspace{-0.15in}
\caption{PSNR versus encoding rate comparing proposed M-frame with D-frame and SP-frame for sequences for sequences \texttt{Balloons}, \texttt{Kendo}, \texttt{Lovebird1} and \texttt{Newspaper} in worst case.}
\label{fig:S3WS}
\end{figure}


The coding results of dynamic view switching for average case and worst case are shown in Fig.~\ref{fig:S3AV} and \ref{fig:S3WS} respectively. BD-rate comparison for average case and worst case can be found in Table~\ref{tab:S3}. From Table~\ref{tab:S3} we observe that our proposed RD-optimized M-frame achieves 57.5\% BD-rate reduction compared to D-frame and 19.3\% BD-rate reduction compared to SP-frame. From Table~\ref{tab:S3} we observe that our proposed RD-optimized M-frame achieves 58.7\% BD-rate reduction compared to D-frame and 36.4\% BD-rate reduction compared to SP-frame.

\section{Conclusion}
\label{sec:conclude}

In this paper, we propose a new merging operator---piecewise constant (PWC) function---for merging different reconstructed versions of a target frame to a unique one---to enable stream switching while preserving coding efficiency. Specifically, in order to merge $k$-th transform coefficients of different side information (SI) frames to the same value, we encode appropriate step sizes and horizontal shift parameters of a \texttt{floor} function, so that all the SI coefficients fall on the same function step. We propose two methods to select \texttt{floor} function parameters for signal merging. In the first method, we selected parameters so that coefficients are merged identically to a pre-determined target value. In the second method, the merged target value can be RD-optimized to induce better coding performance. Experimental results show that for both cases, our proposed merge frame has significant coding gain over an implementation of DSC frame and H.264 SP-frames with a reduction in decoder complexity.





\bibliographystyle{IEEEtran}
\bibliography{ref,ref2}
\end{document}